\documentclass[review,onefignum,onetabnum]{siamart190516}

\usepackage{graphics}
\usepackage{graphicx}
\graphicspath{{./fig/}}
\usepackage{color}
\usepackage{graphicx,epstopdf} 
\usepackage[caption=false]{subfig}
\usepackage{amssymb}
\usepackage{amsmath}
\usepackage{mathtools}
\usepackage{amsfonts}
\usepackage{ntheorem}

\usepackage{amsmath}
\usepackage{enumitem} 
\usepackage{mathrsfs}
\usepackage{latexsym, bm}
\usepackage{hyperref}
\hypersetup{
  colorlinks   = true, 
  urlcolor     = orange, 
  linkcolor    = red, 
  citecolor   = green 
}

\usepackage{pdfpages}
\newtheorem{remark}{Remark}

\newcommand{\interior}{\operatorname{int}}

\usepackage{url}
 \DeclareMathOperator{\Span}{span}
\DeclareMathOperator{\diag}{diag}
\usepackage{lineno}

\newcommand{\until}[1]{\{1,\dots, #1\}}

\newcommand{\setdef}[2]{\{#1 \; | \; #2\}}

\newcommand{\map}[3]{#1: #2 \rightarrow #3}

\newcommand\oprocendsymbol{\hbox{$\square$}}
\newcommand\oprocend{\relax\ifmmode\else\unskip\hfill\fi\oprocendsymbol}

\renewcommand{\theenumi}{(\roman{enumi})}

\DeclareSymbolFont{bbold}{U}{bbold}{m}{n}
\DeclareSymbolFontAlphabet{\mathbbold}{bbold}

\newcommand{\Prob}{\mathbb{P}}
\newcommand{\E}{\mathbb{E}}
\newcommand{\Ave}{\textup{ave}}
\newcommand{\Var}{\textup{Var}}
\newcommand{\argmin}{\textup{argmin}}
\newcommand{\col}{\textup{col}}
\allowdisplaybreaks[2]

\usepackage[prependcaption,colorinlistoftodos]{todonotes}

\usepackage{hyperref}

\usepackage{algorithm}
\usepackage[noend]{algorithmic}
\algsetup{indent=2em}

\renewcommand{\algorithmicrequire}{\textbf{Input:}}
\renewcommand{\algorithmicensure}{\textbf{Output:}}

\headers{How social influence affects the wisdom of crowds}{Y. Tian, L. Wang, and F. Bullo}

\title{How social influence affects the wisdom of crowds in influence
  networks\thanks{Submitted to the editors DATE.  The authors thank
    Professor Fabio Fagnani, Politecnico di Torino, Italy, for insightful
    early conversations that helped define this work.  \funding{The work of
      Ye Tian and Long Wang was supported by the National Natural Science
      Foundation of China under Grant 62036002. The work of Francesco Bullo was supported in part by
      U. S. Army Research Office under grant W911NF-15-1-0577 and grant W911NF-22-1-0233.{\it(Corresponding author: Long Wang)}}}}

\author{Ye Tian\thanks{Center for Systems and Control, College of Engineering, Peking University, Beijing 100871, China (\email{tinybeta7.1@gmail.com}, \email{longwang@pku.edu.cn}).}
\and Long Wang\footnotemark[2]
\and Francesco Bullo\thanks{Mechanical Engineering Department and the Center of Control, Dynamical Systems, and Computation, University of California at Santa Barbara, Santa Barbara, CA 93106 USA (\email{bullo@engineering.ucsb.edu}).}}

\usepackage{amsopn}
\begin{document}

\maketitle
\begin{abstract}
  A long-standing debate is whether social influence improves the
  collective wisdom of a crowd or undermines it. This paper addresses the
  question based on a na\"{\i}ve learning setting in influence systems
  theory: in our models individuals evolve their estimates of an unknown
  truth according to the weighted-average opinion dynamics. A formal
  mathematization is provided with rigorous theoretical analysis. We obtain
  various conditions for improving, optimizing and undermining the crowd
  accuracy, respectively. We prove that if the wisdom of a finite-size group is improved, then the
  collective estimate converges to the truth as the group size increases. We
  show that whether social influence improves or undermines the wisdom is
  determined by the social power allocations in the influence system: if the
  influence system allocates relatively larger social power to relatively
  more accurate individuals, it improves the wisdom; on the contrary, if
  the influence system assigns less social power to more accurate
  individuals, it undermines the wisdom. At a population level,
  individuals' susceptibilities to interpersonal influence and network
  centralities are both crucial. To improve the wisdom, more accurate
  individuals should be less susceptible and have larger network
  centralities. Particularly, in democratic influence networks, if
  relatively more accurate individuals are relatively less susceptible, the
  wisdom is improved; if more accurate individuals are more susceptible,
  the wisdom is undermined, which is consistent with the reported empirical
  evidence. Our investigation provides a theoretical framework for
  understanding the role social influence plays in the emergence of
  collective wisdom.
\end{abstract}

\section{Introduction}\label{S1}
\paragraph*{Problem description and motivation}
The wisdom of crowds effect refers to the phenomenon of improvement in estimate accuracy by pooling many independent estimates. Numerous experiments and simulations have been contributed to understanding and explaining this collective intelligence for decades. A very recent interest and debate with respect to social networks is the question of whether and, if so, how social influence improves or undermines the wisdom of crowds. As much research with empirical evidence, even in support of contrary conclusions, was reported, there is still a lack of rigorous mathematical formulation and analysis.

On the other hand, influence system theory studies the dissemination and aggregation of opinions in influence networks based on mathematical models of opinion dynamics. A relevant model is the so-called na\"{\i}ve learning model. In this model, individuals update their estimates for an unknown truth according to the French-DeGroot (FD) opinion dynamics with their initial estimates disturbed by independent zero-mean noise. Thereby, the group is initially wise in the sense that the initial collective estimate converges to the truth as the group size increases. However, existing results focus on how to preserve the wisdom as the group size tends to infinity~\cite{BG-MOJ:10, FB-FF-BF:17l}, instead of improving or undermining the wisdom.

This paper addresses the question of how social influence improves or undermines the wisdom of finite-size groups in a general framework of weighted-average opinion dynamics. Before the social influence process, individuals independently generate their initial estimates for an unknown truth according to their personal knowledge or expertise, which are described by a family of independent random variables with uniformly unbiased expectations and possibly different variances. Then, social influence takes effect and individuals evolve their estimates under the influence of each other. We propose the formal formulation of wisdom of crowds as well as definitions of improving, optimizing and undermining of the wisdom for a class of weighted-average opinion dynamics. Necessary and/or sufficient conditions for improving, optimizing and undermining of wisdom are derived. Then our theoretical results are applied to the FD opinion dynamics to examine empirical findings in the literature. We aim to provide a theoretical framework to understand what social processes do indeed promote or diminish the collective intelligence and how to organize a wiser group.

\paragraph*{Literature review}
As well-known as the idiom ``two heads are better than one", wisdom of
crowds has been a source of fascination to academia for
centuries~\cite{JSU:04}. A comparable concept in the area of animal
behavior is called the many wrongs principle~\cite{AMS:04}. In 1907,
Galton~\cite{FG:1907} reported an experiment in which averaging hundreds of
guesses for the weight of an ox yielded a collective estimate essentially
close to the true value. In the middle of the last century, a structured
communication technique, known as the Delphi method, was developed by the RAND Corporation to systematically exploit experts'
opinions~\cite{ND-OH:63}. Lots of experiments have been conducted on the
Delphi method; evidence shows that the Delphi method produces better
decisions than uncontrolled group discussion across many
domains~\cite{RJB:74}. The most significant finding of the Delphi experiments
is that participants without strong convictions tend to change their
estimates, while those who feel they have a good argument for a deviant
estimate tend to retain and defend their original
estimates~\cite{OH:67}. In~\cite{NCD:68}, participants who did not change
their estimates are called holdouts, in contrast to swingers who changed their estimates. The holdouts were reported more
accurate than swingers and even than the entire group. The same phenomenon
was also reported in~\cite{GM-GGdP:15}and~\cite{JB-DB-DC:17} recently, while~\cite{GM-GGdP:15} named it the
\emph{Parent\'{e} and Anderson-Parent\'{e} (PAP) hypothesis}.

In a related area, the investigation of social networks focuses on the
formation and evolution of individuals' opinions and social power in the
presence of interpersonal influence. Classic opinion dynamics models
include the FD model~\cite{MHDG:74}, the Friedkin-Johnsen
model~\cite{NEF-ECJ:99, YT-LW:18}, the Hegselmann-Krause
model~\cite{VDB-JMH-JNT:09,VDB-JMH-JNT:10} and the Altafini
model~\cite{CA:13,DM-ZM-YH:19}, to name but a
few~\cite{DB-HT-TB:16,AVP-RT:17}. In this literature, the FD
model serves as a basis of others and is widely established. As the
research progresses, social power and network centrality turn out to be
crucial in determining the outcome of opinion
formation~\cite{PJ-AM-NEF-FB:13d,YT-PJ-AM-LW-NEF-FB:19c}, where social
power indicates the dominance of an individual's initial opinion on group's
final opinions, and network centrality measures the relative importance of
an individual in the network.

Naturally, there arises the question of whether social influence improves
or undermines the wisdom, and how. Lorenz et al.~\cite{JL-HR-FS-DH:11}
concluded, with empirical evidence, that even mild social influence
undermines the wisdom by diminishing the diversity of estimates. In
contrast, Becker et al.~\cite{JB-DB-DC:17} showed that social influence
improves the accuracy of collective estimate even as individuals' estimates
become more similar. In~\cite{AEM-JBS-RPL:14} and~\cite{GM-GGdP:15}, the
authors showed, respectively, that selecting individuals according to their
exerted average accuracy and their resistance to social influence makes
large improvements in group accuracy. In~\cite{AA-ANC-AA-PMK-MM-AP:20}, the
authors showed that social influence produces more accurate collective
estimates in the presence of plasticity and feedback.

\paragraph*{Contributions} 
This paper investigates the effect of social influence on wisdom of crowds in influence networks. Different from the na\"{\i}ve learning model, we address the question of how social influence improves or undermines the wisdom in finite-size groups. First, we propose mathematical formulation for wisdom of crowds and formal definitions for improving, optimizing and undermining of the wisdom in influence networks modelled by weighted-average opinion dynamics. We show that in our definitions, if the influence system improves the wisdom of a finite-size group, the collective estimate converges to the truth in probability as the group size increases.

Second, we study how to improve and to optimize the wisdom. We show that at an influence system level, whether social influence improves or undermines the wisdom is decided by the social power allocations of the influence system. Several notions for consistency of an influence system's social power allocations and individuals' accuracy are defined. Necessary and/or sufficient conditions for improving and optimizing the wisdom are provided. We prove that social influence optimizes the wisdom if and only if the influence system's social power allocations are exactly proportional to individuals' accuracy. The wisdom is improved only if individuals' variances are not uniform; and if so, social influence improves the wisdom if the influence system allocates relatively more, but not too much, social power to more accurate individuals. Moreover, we prove that the select crowd strategy works under proper conditions. We also define a hierarchy of individuals' variances with which all social power allocations in certain orderings improve the wisdom; an algorithm is further designed to find all these orderings for given individuals' variances and its complexity is analyzed.

Third, we study how to undermine the wisdom. We show that if the influence system allocates more social power to less accurate individuals, the wisdom is undermined. Moreover, we prove that if social influence improves the wisdom with all social power allocations in a certain ordering, then social influence undermines the wisdom with all social power allocations in the reverse ordering. Therefore, the aforementioned algorithm also finds all the orderings that all social power allocations in these orderings undermine the wisdom. 

Finally, we apply our theoretical results to the FD opinion dynamics. We show that at a population level, individuals' susceptibilities to interpersonal influence and network centralities both play important roles. In general influence networks, the \emph{PAP hypothesis} is neither sufficient nor necessary to improve the wisdom. Roughly speaking, social influence improves the wisdom if more accurate individuals have larger network centralities and are less susceptible to social influence. However, if the influence network is democratic, the \emph{PAP hypothesis} is, to some degree, sufficient and/or necessary to improve or to optimize the wisdom. We also show that the wisdom can be improved or optimized even in autocratic influence networks. 

Our theoretical analysis reveals some findings of sociological significance and contributes to understanding the role social influence plays in the emergence of collective intelligence. First, social influence can both improve and undermine the wisdom of crowds, which mainly depends on individuals' accuracy and the social power allocations of the influence system. If the influence system assigns relatively larger social power to relatively more accurate individuals, it improves the wisdom; if the influence system assigns less social power to more accurate individuals, it undermines the wisdom. Moreover, there is a symmetric conclusion that if the influence system with a certain ordering of social power allocations improves the wisdom, then it undermines the wisdom with the inversely ordered social power allocations. Lastly, in democratic influence network, if relatively more accurate individuals are relatively less susceptible to social influence, the wisdom is improved, which supports the \emph{PAP hypothesis}. In general influence networks, how social influence takes effect is more complicated and is determined not only by individuals' accuracy and susceptibilities, but also by their network centralities. Simply put, if more accurate individuals have larger network centralities and are less susceptible to social influence, the wisdom is improved.
\paragraph*{Paper organization}
In \cref{S2}, we propose mathematizations for improving, optimizing and undermining of the wisdom in influence networks. \Cref{S3,S4} investigate how social influence improves, optimizes and undermines the wisdom at influence system level. In \cref{S5} we further apply our results to the FD opinion dynamics. \Cref{S6} concludes the paper.  

\paragraph*{Notation} 
$\mathbf{1}_{n}$ and $I_{n}$ denote the $n\times 1$ all-ones vector and the $n\times n$ identity matrix, respectively. $\mathbf{e}_{i}$ denotes the $i$-th standard basis vector with proper dimension. $\mathbb{R}$, $\mathbb{R}^{n}$ and $\mathbb{R}^{n\times n}$ denote, respectively, the sets of real number, $n$-dimensional real vector and $n\times n$ real matrix. For vector or square matrix $A^{(n)}$, the superscript $(n)$ denotes its dimension if it is necessary to clarify. Given $\delta\in\mathbb{R}^{n}$, $[\delta]=\diag(\delta)$ denotes a diagonal matrix with diagonal elements $\delta_{1}, \dots, \delta_{n}$. The $n$-simplex is denoted by $\Delta_{n}\!=\!\setdef{z\in\mathbb{R}^{n}}{z\geq 0, \mathbf{1}^{\top}_{n}z\!=\!1}$; $\interior{\Delta_{n}}\!=\!\setdef{z\in\mathbb{R}^{n}}{z> 0, \mathbf{1}^{\top}_{n}z\!=\!1}$ denotes its interior. A nonnegative matrix is row-stochastic (column-stochastic) if its row (column) sums are $1$; it is doubly-stochastic if both its row and column sums are $1$. The weighted digraph $\mathcal{G}(W)$ associated with nonnegative matrix $W$ is defined as: the node set is $\until n$; there is a directed edge $(i,j)$ from nodes $i$ to $j$ if and only if $W_{ij}>0$. $\mathcal{G}(W)$ is a star topology if all its directed edges are either from or to a center node. A strongly connected component (SCC) of $\mathcal{G}(W)$ is a maximal strongly connected subgraph. A SCC is called a sink SCC if there exists no directed edge from this SCC to others. For an irreducible matrix $W$, $\mathcal{G}(W)$ is called democratic if $W$ is doubly-stochastic; $\mathcal{G}(W)$ is called autocratic if $\mathcal{G}(W)$ is a star topology. 
\section{Mathematical formulation of the wisdom of crowds in influence networks}\label{S2}
\subsection{The wisdom of crowds without social influence}

Consider $n\geq 2$ individuals in an influence network interacting their
estimates for an unknown state with constant true value
$\mu\in\mathbb{R}$. Suppose that individuals do not know the exact value of $\mu$, but have fragmented knowledge or
clues to estimate it. Before interacting with others, each individual $i$ gives an initial estimate $y_{i}(0)$, independently, according to the knowledge or clues it possesses. Specifically, $y_{i}(0)$ is a random variable with expectation $\E[y_{i}(0)]=\mu$ and variance $\Var[y_{i}(0)]=\sigma_{i}^{2}>0$. $y_{1}(0),\dots,y_{n}(0)$ are independent since the initial estimates are generated independently by individuals. In this paper,
we employ the assumption in~\cite{FG:1907} that the collective estimate,
denoted by $y_{\col}$, is aggregated by averaging all individuals'
estimates arithmetically. Let
\begin{equation}
y_{\col}(0)=\Ave(y(0))=\frac{1}{n}\sum_{i=1}^{n}y_{i}(0)
\end{equation}
denote the initial collective estimate. Due to the independence of $y_{1}(0),\dots,y_{n}(0)$, we have $\E[y_{\col}(0)]=\mu$ and $\Var[y_{\col}(0)]=\frac{1}{n^2}\sum_{i=1}^{n}\sigma_{i}^{2}$. Since $\E[y_{i}(0)]=\E[y_{\col}(0)]=\mu$, a smaller variance indicates that the estimate tends to be closer to the truth. Denote $\sigma^{2}=(\sigma_{1}^{2},\dots,\sigma_{n}^{2})^{\top}$, $\sigma_{\max}^{2}=\max_{i}\sigma_{i}^{2}$ and $\sigma_{\min}^{2}=\min_{i}\sigma_{i}^{2}$. Note that $\Var[y_{\col}(0)]<\Ave(\sigma^2)$, and $\Var[y_{\col}(0)]<\sigma_{\min}^{2}$ if $\sigma_{\max}^{2}<n\sigma_{\min}^{2}$. That is, the variance of the initial collective estimate is smaller than the average variance of individuals' initial estimates and can be smaller than the variance of any single individual's initial estimate. This condensation of variance exhibits that simply averaging a group of individuals' independent estimates yields a better collective estimate which even outperforms the estimate of the best individual, known as the wisdom of crowds effect.   

\begin{remark}
\begin{enumerate}
\item The wisdom of crowds problem usually focuses on tasks for which individuals are unlikely to know the exact truth but are not clueless either \cite{JL-HR-FS-DH:11}. Here, we assume $\E[y_{i}(0)]=\mu$ to avoid the case that individuals' independent estimates are uninformative at all. 
\item We do not specify a particular distribution for each individual's initial estimate. It is possible that two individuals' initial estimates have different distributions, e.g., one is normally distributed and the other is binomially distributed. Intuitively, this depends on the nature of the task and individuals' local knowledge, as well as how they exploit the knowledge.
\item A direct measurement of the wisdom is the error of the collective estimate, such as the square error and the normalized absolute error respectively used in \cite{JL-HR-FS-DH:11} and \cite{JB-DB-DC:17}. In this paper, since the initial collective estimate is unbiased, we use its variance, i.e., the mean square error, to assess the level of wisdom. 
\end{enumerate}
\end{remark}

\subsection{Improving, optimizing and undermining of the wisdom in influence networks}
Suppose that after execution of a social influence process, individuals' estimates converge and are given by 
\begin{equation}\label{eq1}
\lim_{k\to\infty}y(k)=Vy(0),
\end{equation} 
where $V$ is the row-stochastic transition matrix of a weighted-average opinion dynamics. In other words, each individual's final estimate $\lim_{k\to\infty}y_{i}(k)=\sum_{j=1}^{n}V_{ij}y_{j}(0)$ is the weighted average of all individuals' initial estimates. Then, the final collective estimate is
\begin{align*}
\lim_{k\to\infty}y_{\col}(k)=\frac{1}{n}\sum_{i=1}^{n}\sum_{j=1}^{n}V_{ij}y_{j}(0)=\sum_{i=1}^{n}x_{i}y_{i}(0),
\end{align*}
where $x_{i}=\frac{1}{n}\sum_{j=1}^{n}V_{ji}$ indicates individual $i$'s social power exerted over the influence process and $x\in\Delta_{n}$ is the social power allocation of the influence system, as defined in~\cite{NEF:11}. Similarly, we have $\E[\lim_{k\to\infty}y_{\col}(k)]=\mu$ and $\Var[\lim_{k\to\infty}y_{\col}(k)]=\sum_{i=1}^{n}x_{i}^{2}\sigma_{i}^{2}$. Since $\E[\lim_{k\to\infty}y_{\col}(k)]=\E[y_{\col}(0)]$, $\Var[\lim_{k\to\infty}y_{\col}(k)]<\Var[y_{\col}(0)]$ means that the influence process condenses the collective variance thus improves the collective estimate, while $\Var[\lim_{k\to\infty}y_{\col}(k)]>\Var[y_{\col}(0)]$ implies that the influence process amplifies the collective variance thereby undermines the collective estimate.
\begin{definition}[Improving, optimizing, undermining of the wisdom and converging to the truth]\label{D1}
For social influence process~\eqref{eq1} with individuals' initial estimates $y_{i}(0)\sim (\mu,\sigma_{i}^{2})$, $i\in\until{n}$ and social power allocation $x\in\Delta_{n}$, we say social influence or system~\eqref{eq1} asymptotically 
\begin{enumerate}

\item improves the wisdom if $\sum_{i=1}^{n}x_{i}^{2}\sigma_{i}^{2}<\frac{1}{n^2}\sum_{i=1}^{n}\sigma_{i}^{2}$; \label{D1-1} 

\item optimizes the wisdom if $\sum_{i=1}^{n}x_{i}^{2}\sigma_{i}^{2}=\min_{z\in\Delta_{n}}\sum_{i=1}^{n}z_{i}^{2}\sigma_{i}^{2}$; \label{D1-2} 

\item undermines the wisdom if $\sum_{i=1}^{n}x_{i}^{2}\sigma_{i}^{2}>\frac{1}{n^2}\sum_{i=1}^{n}\sigma_{i}^{2}$.

\setlength\parindent{-2em} Additionally, we say 

\item the final collective estimate asymptotically converges to the truth in probability ({\it i.p.}) as the group size increases if $\lim_{n\to\infty}\Prob[\lim_{k\to\infty}y^{(n)}_{\col}(k)=\mu]=1$. \label{D1-3}

\end{enumerate}    
\end{definition}

\begin{remark}
Social influence process \cref{eq1} can model all the opinion dynamics where individuals' opinions converge to convex combinations of their initial opinions. Therefore, \cref{D1} is well-posed for influence processes such as the FD model, the Friedkin-Johnsen model and the Hegselmann-Krause model, etc. 
\end{remark}

In \cref{D1}, improving the wisdom means that the variance of the final collective estimate aggregated by the social influence process is smaller than the variance of the initial collective estimate obtained by averaging individuals' initial estimates. Additionally, optimizing the wisdom means that the variance of the final collective estimate achieves the minimum for given variances of individuals' initial estimates.
\subsection{Connections between improving the wisdom and converging to the truth}
The next lemma states that~\cref{D1}~\ref{D1-1} implies~\cref{D1}~\ref{D1-3} if individuals' variances are uniformly upper bounded as the group size increases.
\begin{lemma}[Improving the wisdom implies converging to the truth]\label{L1} 
For social influence process~\eqref{eq1}, suppose that there exists $\beta>0$ such that $\sigma^{2}_{i}\leq\beta<\infty$ for all $i\in\until{n}$. If social influence asymptotically improves the wisdom, then the final collective estimate asymptotically converges to the truth {\it i.p.} as $n\to \infty$.
\end{lemma}

\Cref{L1} is proved in \cref{AP1}. In~\cref{L1}, if individuals' estimates achieve consensus, then each individual's final estimate also converges to the truth {\it i.p.} as $n\to\infty$. In~\cite{BG-MOJ:10}, it has been proved that for the FD model, if the influence matrix sequence satisfies balance and minimal out-dispersion conditions, individuals' final estimates converge to the truth {\it i.p.} as $n\to\infty$. \Cref{L1} provides a different condition from the perspective of condensing the collective variance of finite-size groups. Therefore, all results on improving the wisdom of finite-size groups imply the collective estimate converges to the truth as the group size increases. To avoid trivialities, we shall not state the results about converging to the truth in the sequel.

\section{The role of social power in improving and optimizing the wisdom}\label{S3}
\subsection{Consistency notions of influence systems}
According to~\cref{D1}, whether an influence system improves or undermines the wisdom depends exclusively on the relation between its social power allocations and individuals' variances. Given individuals' variances $\sigma^{2}\in\mathbb{R}^{n}$, define $\map{\mathcal{E}_{\sigma^{2}}}{\Delta_{n}}{\mathbb{R}_{>0}}$ by $\mathcal{E}_{\sigma^{2}}(z)=\sum_{i=1}^{n}z_{i}^{2}\sigma_{i}^{2}$, then $\mathcal{E}_{\sigma^{2}}(x)$ is the final collective variance of system \eqref{eq1} with social power allocation $x$. Hence, by~\cref{D1}~\ref{D1-1}, $\mathcal{A}_{\sigma^{2}}=\setdef{z\in \Delta_{n}}{\mathcal{E}_{\sigma^{2}}(z)<\mathcal{E}_{\sigma^{2}}(\mathbf{1}_{n}/n)}$ represents the improvement region, i.e., the wisdom is improved if and only if $x\in\mathcal{A}_{\sigma^{2}}$. In fact, $\mathcal{A}_{\sigma^{2}}$ is the interior of the intersection of the $n$-simplex $\Delta_{n}$ and the hyperellipsoid 
\begin{align*}
\sum_{i=1}^{n}\frac{z_{i}^2}{\frac{\sum_{j=1}^{n}\sigma_{j}^{2}}{n^{2}\sigma_{i}^{2}}}=1,
\end{align*}
see \cref{fig1}. Intuitively, to improve the wisdom an influence system needs to allocate larger social power to more accurate individuals, that is, the ordering of its social power allocations should be consistent with the ordering of individuals' accuracy, where individual $i$'s accuracy is indicated by $1/\sigma^{2}_{i}$. However, \cref{f1-a,f1-b} show that this is neither sufficient nor necessary to improve the wisdom. Based on this observation, we define the following notions of consistency progressively.

\begin{figure}[t]
\centering
 \subfloat[t][$\sigma^{2}=(1, 2, 3)^{\top}$]{
\begin{minipage}[t]{0.35\textwidth}\label{f1-a}
 \centering
  \includegraphics[width=\hsize]{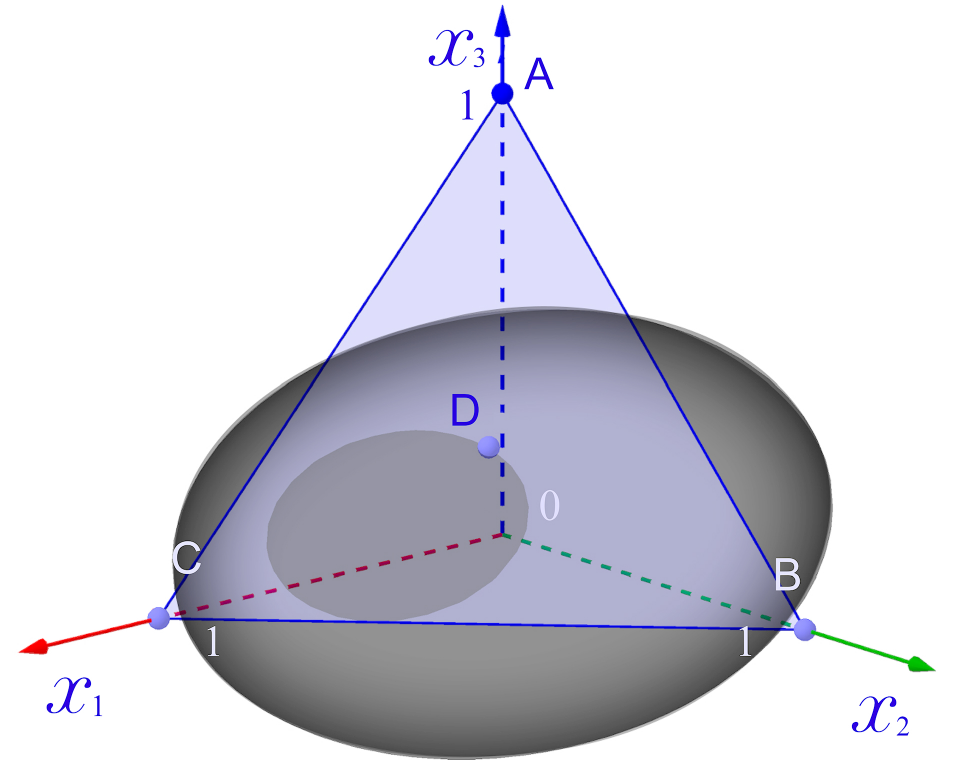}\\
 \end{minipage}
}
\hspace{4ex}
\subfloat[t][$\sigma^{2}=(1, 3, 4)^{\top}$]{
\begin{minipage}[t]{0.35\textwidth}\label{f1-b}
 \centering
  \includegraphics[width=\hsize]{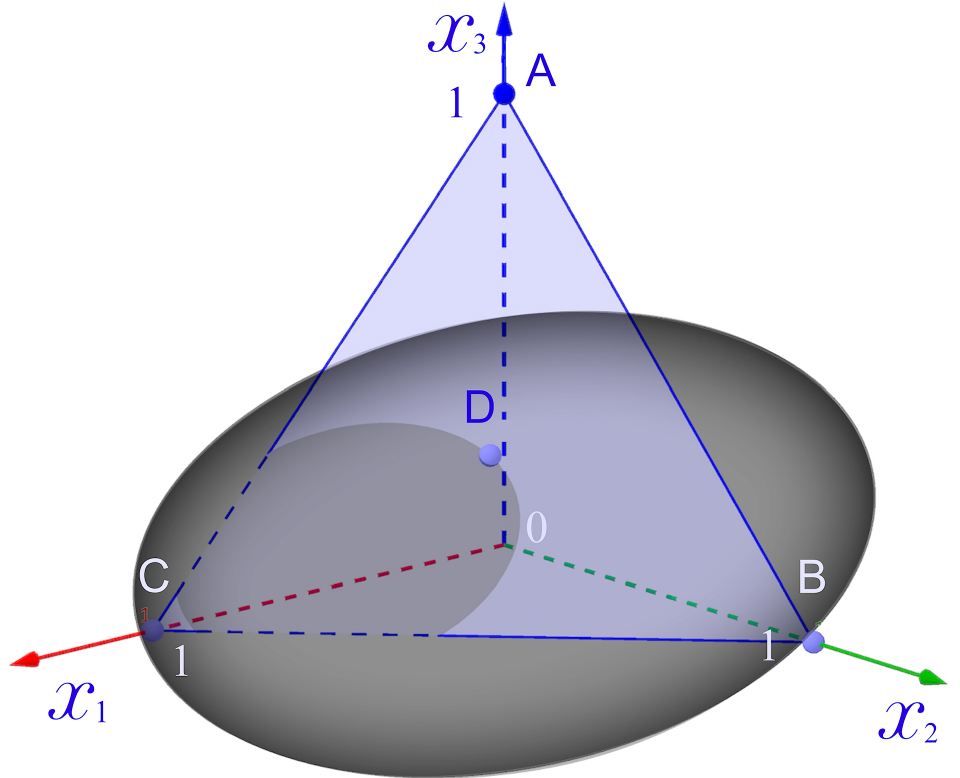}\\
  \end{minipage}
}
\vspace{-2ex}

\subfloat[t][$\sigma^{2}=(1, 4, 9)^{\top}$]{
\begin{minipage}[t]{0.35\textwidth}\label{f1-c}
 \centering
  \includegraphics[width=\hsize]{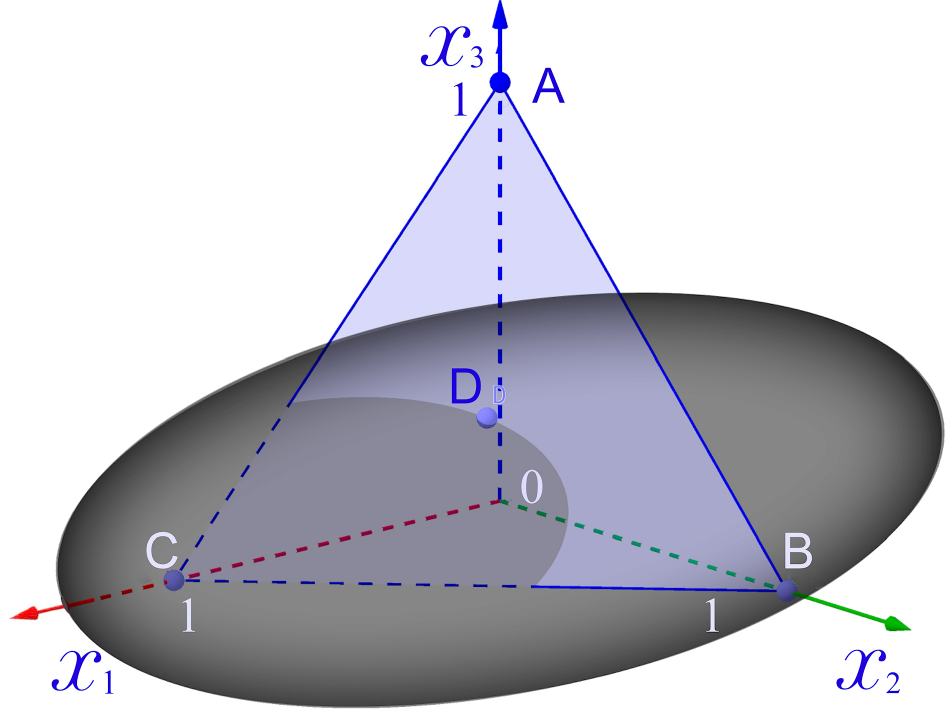}\\
 \end{minipage}
}
\hspace{4ex}
\subfloat[t][$\sigma^{2}=(2, 1, 16)^{\top}$]{
\begin{minipage}[t]{0.375\textwidth}\label{f1-d}
 \centering
  \includegraphics[width=\hsize]{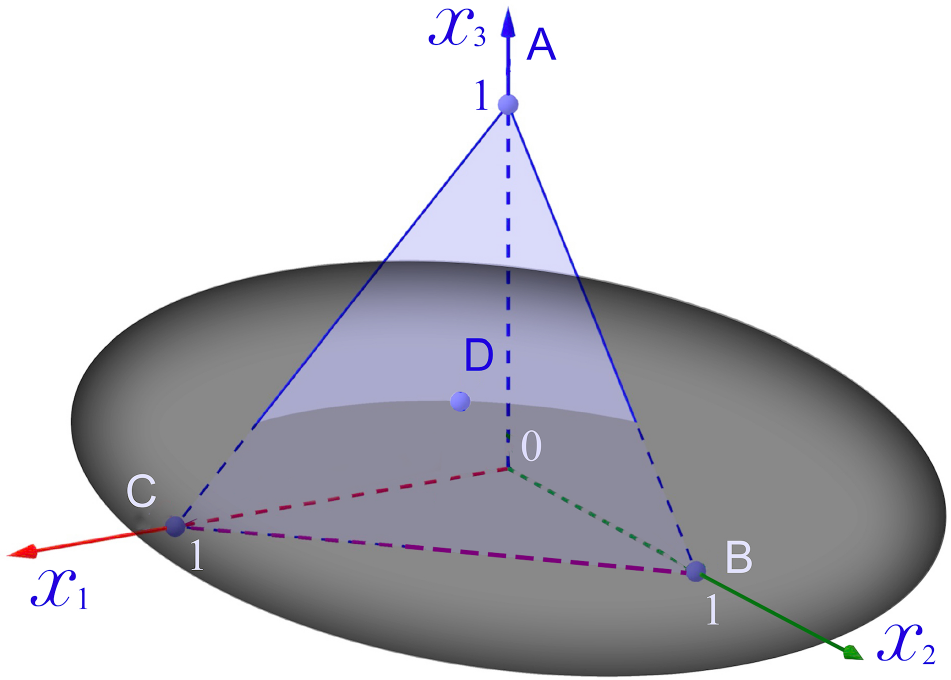}\\
 \end{minipage}
}
\caption{The improvement regions $\mathcal{A}_{\sigma^{2}}$ for different $\sigma^{2}$ depicted in $\Delta_{3}$. For each $\sigma^{2}$, $\mathcal{A}_{\sigma^{2}}$ is the interior of the intersection of $\Delta_{3}$ and the ellipsoid depending on $\sigma^{2}$; if $\sigma^{2}\in\Span{\mathbf{1}_{n}}$, the ellipsoid becomes a ball whose intersection with $\Delta_{3}$ is the point $D=\mathbf{1}_{n}/3$.}\label{fig1}
\end{figure}

\begin{definition}[Consistency notions of influence systems]\label{D2}
Let $x_{i}$ and $1/\sigma^{2}_{i}$ be individual $i$'s social power and accuracy, respectively. We say system~\eqref{eq1} is 
\begin{enumerate}
\item ordering-consistent if $x_{i}\leq x_{j}$ for all $\sigma^{2}_{i}\geq\sigma^{2}_{j}$ and $x=\mathbf{1}_{n}/n$ only if $\sigma^{2}\in\Span{\mathbf{1}_{n}}$;
\item gap-consistent if it is ordering-consistent and $x_{i}\sigma^{2}_{i}\geq x_{j}\sigma^{2}_{j}$ for all $\sigma^{2}_{i}\geq\sigma^{2}_{j}$; \label{D2-3}
\item maximally consistent if $x\in\interior{\Delta_{n}}$ and $\frac{x_{i}}{x_{j}}=\frac{\sigma^{2}_{j}}{\sigma^{2}_{i}}$ for all $i,j\in\until{n}$.
\end{enumerate} 
\end{definition}
\begin{remark}\label{r3}
In \cref{D2}~\ref{D2-3}, if there exist $i,j$ such that $x_{j}>x_{i}=0$, then we have $x_{j}\sigma_{j}^{2}\leq x_{i}\sigma_{i}^{2}=0$, which is contradicted with $x_{j}\sigma_{j}^{2}>0$. Hence, $x\in\interior{\Delta_{n}}$. That is, both gap-consistency and maximal consistency require that the influence system allocates strictly positive social power to all individuals. 
\end{remark}

Ordering-consistency requires that the influence system allocates more social power to more accurate individuals, gap-consistency additionally excludes the case that it assigns too much social power to more accurate individuals, while maximal consistency demands that its social power allocation is exactly proportional to individuals' accuracy. 
\subsection{The effects of gap-consistency and maximal consistency}\label{s3.2}
As shown in~\cref{fig1}, the improvement region $\mathcal{A}_{\sigma^{2}}$ is the interior of the intersection of the $n$-simplex and a hyperellipsoid. The next lemma captures the properties of $\mathcal{A}_{\sigma^{2}}$. 
\begin{lemma}[Properties of the improvement region]\label{L3}
For $\mathcal{E}_{\sigma^{2}}(z)=\sum_{i=1}^{n}z_{i}^{2}\sigma_{i}^{2}$ and $\mathcal{A}_{\sigma^{2}}=\setdef{z\in \Delta_{n}}{\mathcal{E}(z)<\mathcal{E}(\frac{\mathbf{1}_{n}}{n})}$ defined above, let $\mathcal{\tilde{A}}_{\sigma^{2}}=\setdef{z\in\interior{ \Delta_{n}}\setminus \{\frac{\mathbf{1}_{n}}{n}\}}{1\leq\frac{z_{i}}{z_{i+1}}\leq \frac{\sigma_{i+1}^{2}}{\sigma_{i}^{2}}, i\in \until{n-1}}$. The following statements hold:
\begin{enumerate}
\item $\mathcal{E}_{\sigma^{2}}(z)$ is strictly convex on $\Delta_{n}$ and consequently $\mathcal{A}_{\sigma^{2}}$ is convex; \label{L3-0}
\item $\mathcal{A}_{\sigma^{2}}\neq\emptyset$ if and only if $\sigma^{2}\notin\Span{\mathbf{1}_{n}}$; \label{L3-1}
\item $\mathcal{\tilde{A}}_{\sigma^{2}}\subset\mathcal{A}_{\sigma^{2}}$. \label{L3-2}
\end{enumerate}
\end{lemma}

\Cref{L3} is proved in~\cref{AP2}. \Cref{L3}~\ref{L3-0} shows that $\mathcal{E}_{\sigma^{2}}(z)$ is strictly convex on $\Delta_{n}$, which implies that if there exist $l\geq 2$ and $x^{i}\in\Delta_{n}$ for $i\in\until{l}$ such that system~\eqref{eq1} with social power allocation $x^{i}$ improves the wisdom, then for all $\omega\in\Delta_{l}$, system~\eqref{eq1} with social power allocation $\sum_{i=1}^{l}\omega_{i}x^{i}$ improves the wisdom.
\begin{theorem}[Consistent influence systems improve/optimize the wisdom]\label{T2}
For social influence process~\eqref{eq1} with social power allocation $x\in\Delta_{n}$ and individuals' variances $\sigma^{2}\in\mathbb{R}^{n}$, the following statements hold:
\begin{enumerate}
\item if system \eqref{eq1} asymptotically improves the wisdom, then $\sigma^{2}\notin\Span{\mathbf{1}_{n}}$; \label{T2-1}

\item if $\sigma^{2}\notin\Span{\mathbf{1}_{n}}$ and system~\eqref{eq1} is gap-consistent, then it asymptotically improves the wisdom;\label{T2-2}

\item system \eqref{eq1} asymptotically optimizes the wisdom if and only if it is maximally consistent. And if so, the variance of the final collective estimate achieves the minimum $\mathcal{E}^{*}_{\sigma^{2}}=\frac{1}{\sum_{i=1}^{n}\frac{1}{\sigma_{i}^{2}}}$ with the optimal social power allocation
\begin{align*}
x^{*}=\frac{[\sigma^{2}]^{-1}\mathbf{1}_{n}}{\mathbf{1}^{\top}_{n}[\sigma^{2}]^{-1}\mathbf{1}_{n}}.
\end{align*} \label{T2-3}
\end{enumerate}
\end{theorem}

\begin{proof}
Statement~\ref{T2-1} immediately follows from \cref{L3}~\ref{L3-1}. Regarding~\ref{T2-2}, since $\sigma^{2}\notin\Span{\mathbf{1}_{n}}$ and system~\eqref{eq1} is gap-consistent, by \cref{r3}, $x\in\interior{\Delta_{n}}\setminus \{\frac{\mathbf{1}_{n}}{n}\}$. Let $P\in\mathbb{R}^{n}$ be the permutation matrix such that $\tilde{\sigma}^{2}=P\sigma^{2}$ satisfies $\tilde{\sigma}_{i}^{2}\leq \tilde{\sigma}_{i+1}^{2}$, $i\in\until{n-1}$, by the definition of gap-consistency, $\tilde{x}=Px$ satisfies 
\begin{align*}
1\leq\frac{\tilde{x}_{i}}{\tilde{x}_{i+1}}\leq\frac{\tilde{\sigma}_{i+1}^{2}}{\tilde{\sigma}_{i}^{2}} 
\end{align*} 
for all $i\in\until{n-1}$, which means that $\tilde{x}\in\mathcal{\tilde{A}}_{\tilde{\sigma}^{2}}$. By \cref{L3}~\ref{L3-2}, $\mathcal{\tilde{A}}_{\sigma^{2}}\subset\mathcal{A}_{\sigma^{2}}$ for any given $\sigma^{2}$. Hence, we obtain $\tilde{x}\in\mathcal{A}_{\tilde{\sigma}^{2}}$, which is equivalent to $x\in\mathcal{A}_{\sigma^{2}}$. Therefore, social influence asymptotically improves the wisdom.

Regarding~\ref{T2-3}, define $\map{L}{\Delta_{n}\times \mathbb{R}}{\mathbb{R}}$ by $L(z,\lambda)=\sum_{i=1}^{n}z_{i}^{2}\sigma_{i}^{2}+\lambda (\mathbf{1}^{\top}_{n}z-1)$. Then, the optimal solution $z^{*}$ satisfies $2z_{i}^{*}\sigma_{i}^{2}+\lambda=0$ for all $i\in\until{n}$ and $\sum_{j=1}^{n}z_{j}^{*}=1$.
That is, the minimum collective estimate is achieved by 
\begin{align*}
x_{i}^{*}=\frac{\frac{1}{\sigma_{i}^{2}}}{\sum_{j=1}^{n}\frac{1}{\sigma_{j}^{2}}},
\end{align*} 
which is equivalent to $\frac{x_{i}^{*}}{x_{j}^{*}}=\frac{\sigma_{j}^{2}}{\sigma_{i}^{2}}$ for all $i,j\in\until{n}$. Moreover, $\mathcal{E}^{*}_{\sigma^{2}}=\mathcal{E}_{\sigma^2}(x^{*})=\frac{1}{\sum_{i=1}^{n}\frac{1}{\sigma_{i}^{2}}}$. Thus, social influence asymptotically optimizes the wisdom if and only if system~\eqref{eq1} is maximally consistent. 
\end{proof}

\Cref{T2} suggests that the wisdom can be improved only if individuals' variances are non-uniform; and if so, social influence improves the wisdom if system~\eqref{eq1} allocates relatively larger, but not too much, social power to relatively more accurate individuals; social influence optimizes the wisdom if and only if the social power allocation of system~\eqref{eq1} is exactly proportional to individuals' accuracy. In \cref{fig2}, quadrilaterals $DFOH$ (excluding the points $D$) and points $O$ depict the gap-consistency regions and the maximal consistency points of social power allocations in the $3$-simplex for different distributions of variances.  
\begin{figure}[t]
\centering
 \subfloat[t][$\sigma^{2}=(1, 2, 3)^{\top}$]{
\begin{minipage}[t]{0.35\textwidth}\label{f2-a}
 \centering
  \includegraphics[width=\hsize]{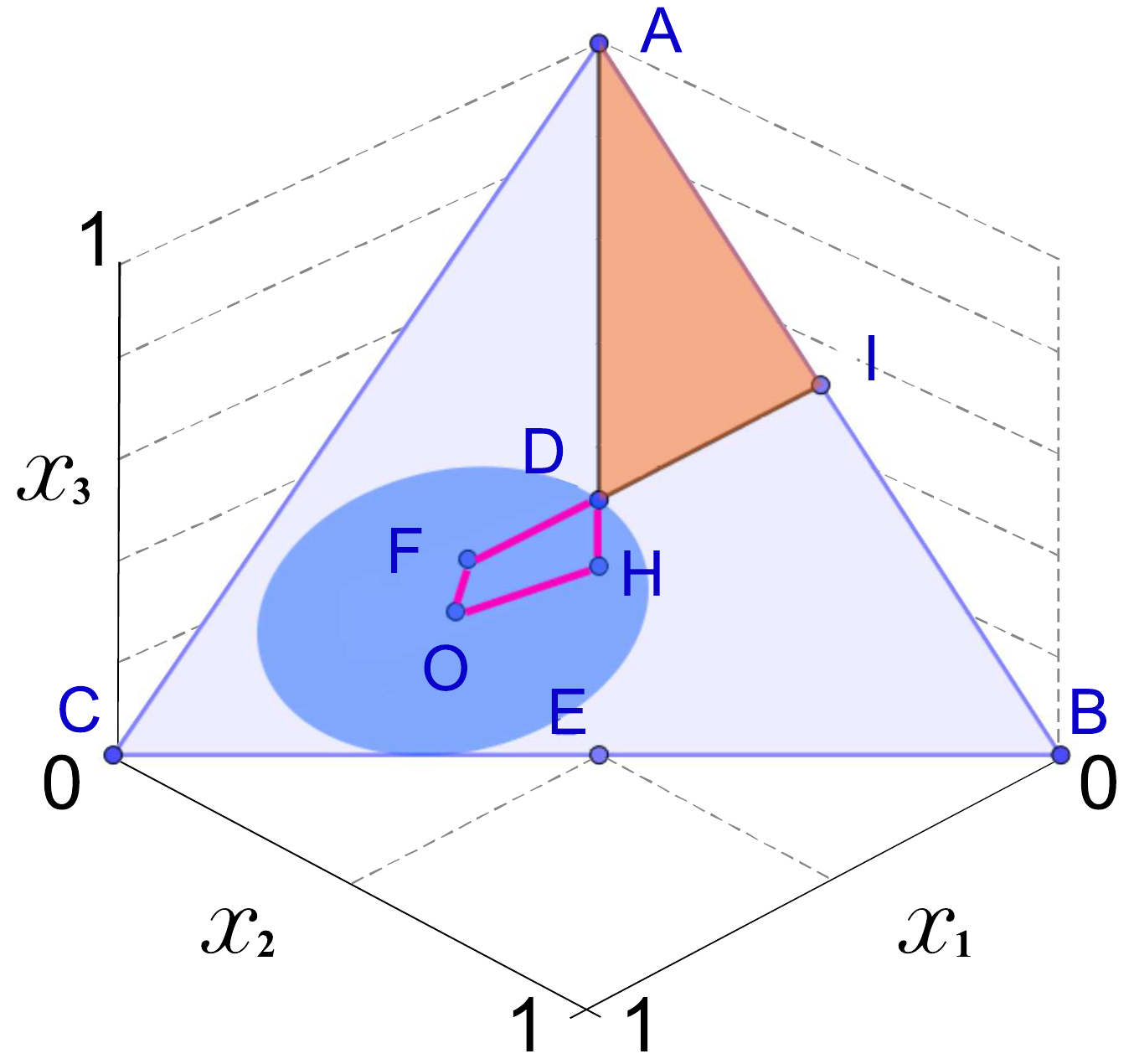}\\
 \end{minipage}
}
\hspace{4ex}
\subfloat[t][$\sigma^{2}=(3, 4, 1)^{\top}$]{
\begin{minipage}[t]{0.35\textwidth}\label{f2-b}
 \centering
  \includegraphics[width=\hsize]{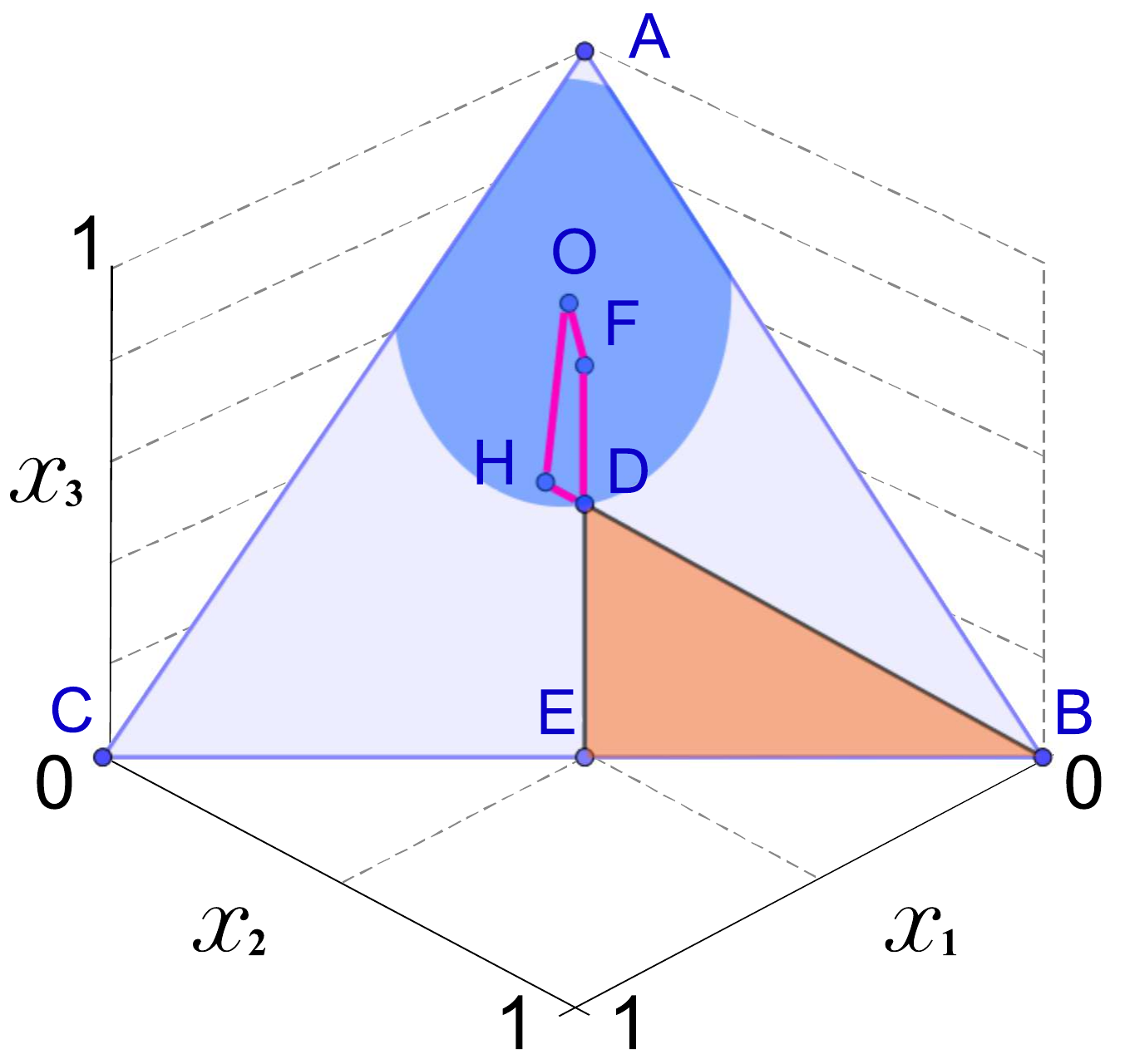}\\
  \end{minipage}
}
\vspace{-2ex}

\subfloat[t][$\sigma^{2}=(1, 4, 9)^{\top}$]{
\begin{minipage}[t]{0.35\textwidth}\label{f2-c}
 \centering
  \includegraphics[width=\hsize]{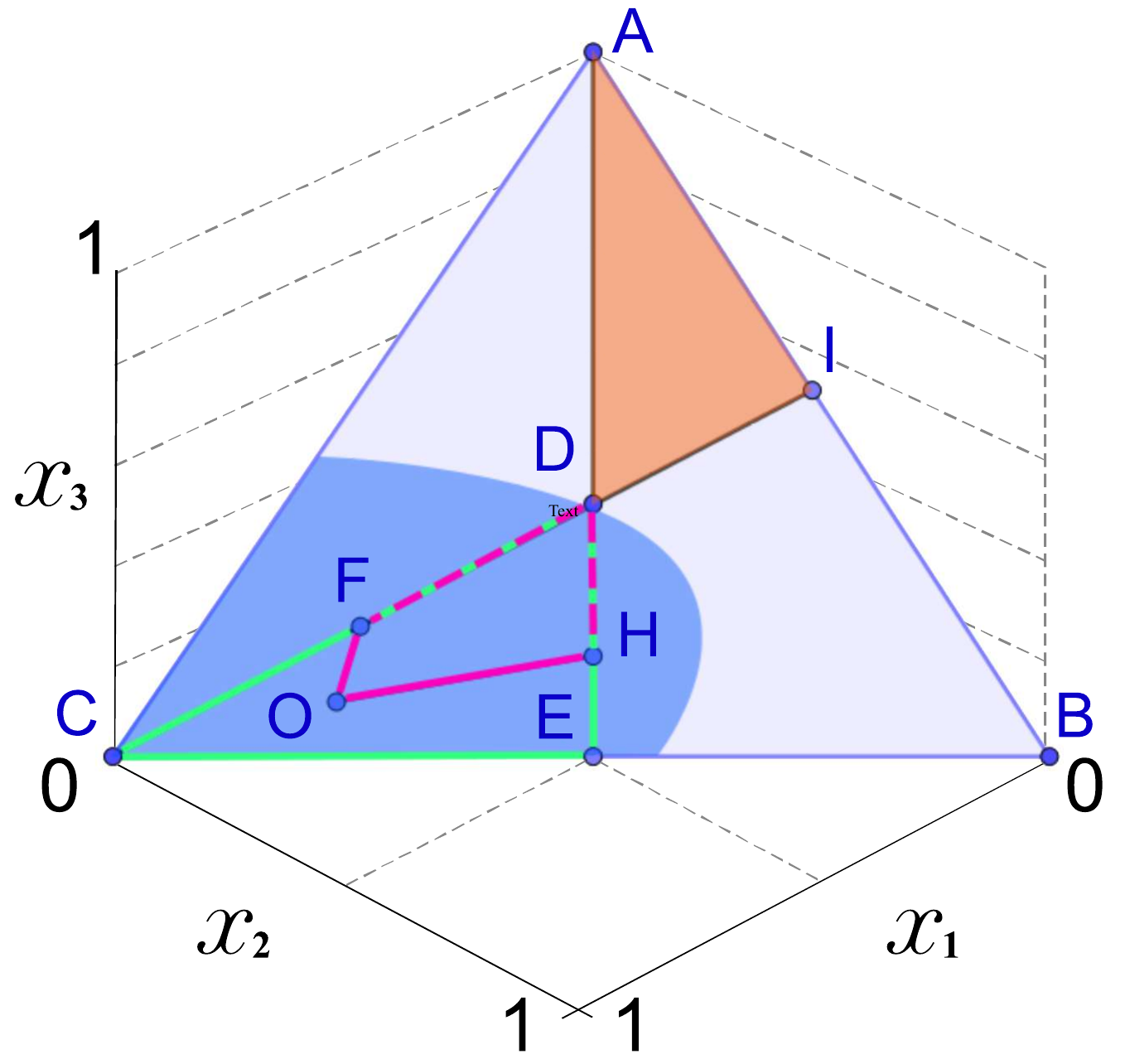}\\
 \end{minipage}
}
\hspace{4ex}
\subfloat[t][$\sigma^{2}=(2, 1, 16)^{\top}$]{
\begin{minipage}[t]{0.35\textwidth}\label{f2-d}
 \centering
  \includegraphics[width=\hsize]{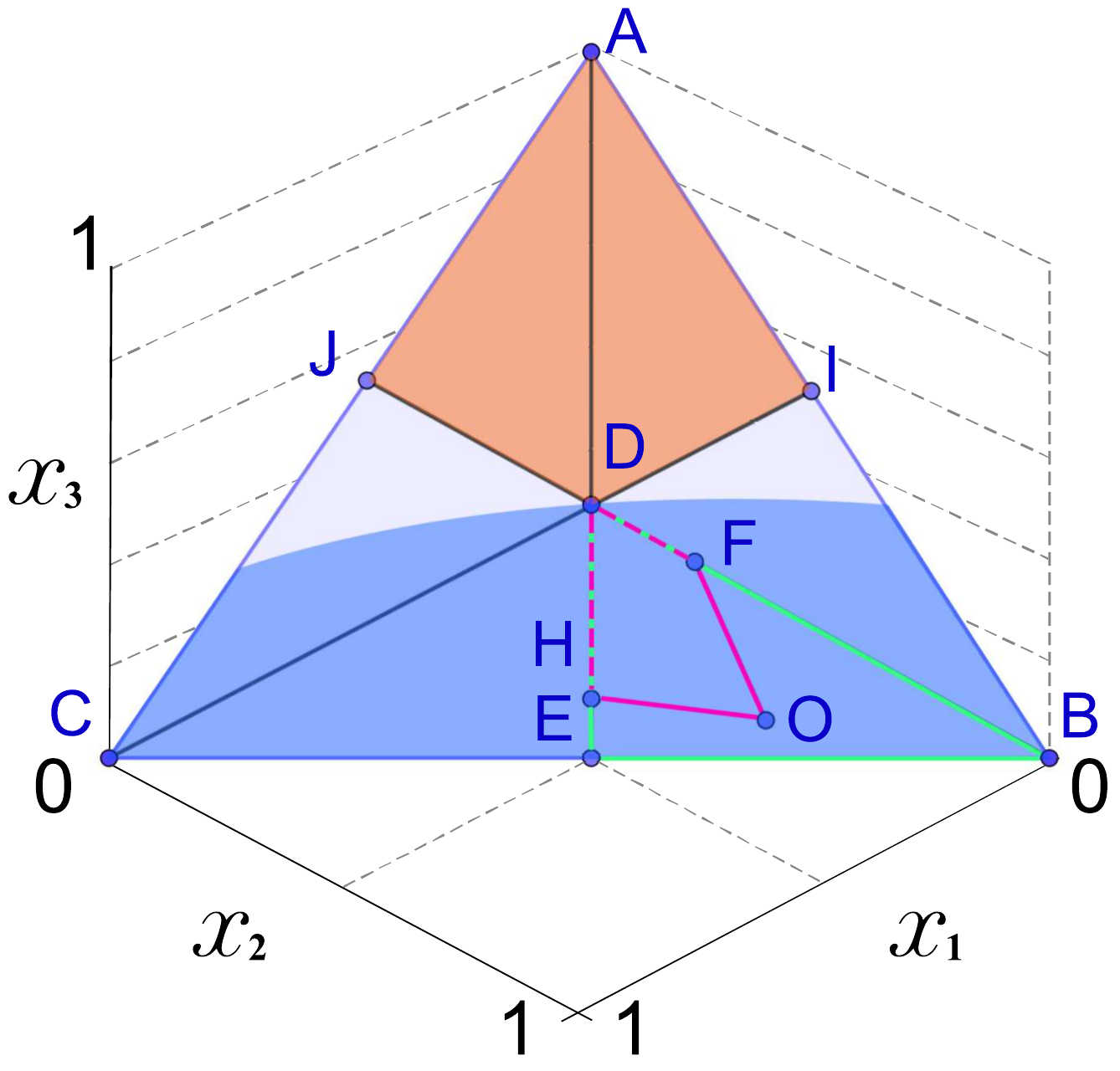}\\
 \end{minipage}
}
\caption{The improvement region, the undermining region and the consistency regions depicted in the $3$-simplex with the point $D=\mathbf{1}_{n}/3$. For each $\sigma^{2}$, blue area is the improvement region, point $O$ is the maximal consistency point, quadrilateral $DFOH$ (excluding point $D$) is the gap-consistency region and the undermining region includes gray and orange areas. The triangles $DCE$ and $DBE$ (excluding the points $D$) in \cref{f2-c,f2-d} are the ordering-consistency regions.}\label{fig2}
\end{figure}

\subsection{The effects of select crowd and hierarchical individuals}
For given individuals' variances $\sigma^{2}\notin\Span{\mathbf{1}_{n}}$, there always exist a gap-consistency region and a maximal consistency point of social power allocation such that the wisdom is improved or optimized. Yet we notice that on one hand, both gap-consistency and maximal consistency require that each individual has strictly positive social power, as discussed in \cref{r3}; on the other hand, whether the influence system is gap-consistent or maximally consistent depends quantitatively upon individuals' social power and variances, while criteria based on the ordering of social power allocations seem more explicit and more intuitive. This subsection first focuses on the effect of select crowd where not every individual is assigned strictly positive social power. Then, a hierarchy of individuals' accuracy is defined with which we can determine the improvement of wisdom by only examining the ordering of individuals' social power.

Note that there exist $\frac{n(n-1)}{2}$ hyperplanes $\{z\in\mathbb{R}^{n}\mid z_{i}=z_{j}\}$, by which the $n$-simplex $\Delta_{n}$ is partitioned into $n!$ hypertriangles, thereby the orderings of components of $z\in\Delta_{n}$ in different hypertriangles are different, see \cref{f2-d}. Let $\tau=(\tau_{1}, \dots,\tau_{n})$ be a permutation of $\tau^{0}=(1, \dots,n)$, $\mathcal{T}$ be the set of all permutations of $\tau^{0}$ and $\Delta^{\tau}_{n}=\setdef{z\in\Delta_{n}\setminus\{ \frac{\mathbf{1}_{n}}{n}\}}{z_{\tau_{i}}\geq z_{\tau_{i+1}}, i\in\until{n-1}}$ be the hypertriangle corresponding to $\tau$, then $\bigcup_{\tau\in\mathcal{T}}\Delta^{\tau}_{n}=\Delta_{n}\setminus\{ \frac{\mathbf{1}_{n}}{n}\}$. For $z\in\mathbb{R}^{n}$, we denote $z_{\tau}=(z_{\tau_{1}},\dots,z_{\tau_{n}})^{\top}$.
 
\begin{lemma}[The properties of the hypertriangle]\label{L4}
For the hypertriangle $\Delta^{\tau^{0}}_{n}=\setdef{z\in \Delta_{n}\setminus \{\frac{\mathbf{1}_{n}}{n}\}}{z_{i}\geq z_{i+1}, i\in \until{n-1}}$ associated with the identity permutation $\tau^{0}=(1, \dots,n)$, let $\mathcal{\hat{A}}_{m}=\setdef{z\in \Delta^{\tau^{0}}_{n}}{z_{m+1}=0}$ with $m<n$, then 
\begin{enumerate}
\item $\Delta^{\tau^{0}}_{n}\subset \mathcal{A}_{\sigma^{2}}$ if and only if  \label{L4-1}
\begin{equation}\label{eq2}
\frac{1}{j^2}\sum_{r=1}^{j}\sigma_{r}^{2}<\frac{1}{n^2}\sum_{i=1}^{n}\sigma_{i}^{2} \quad \text{for all} \quad j\in\until{n-1}; 
\end{equation}

\item $\mathcal{\hat{A}}_{m}\subset \mathcal{A}_{\sigma^{2}}$ if and only if~\eqref{eq2} holds for all $j\in\until{m}$. \label{L4-2}
\end{enumerate}
\end{lemma}

\Cref{L4} is proved in~\cref{AP3}. In practice, social influence systems in which not everyone has strictly positive social power are ubiquitous, such as administrative systems or bureaucracies. The next theorem suggests that in these systems, if more powerful people have smaller upper bounds of variances, in other words, are more expert, the wisdom is improved, regardless of the expertise of people without power.

\begin{theorem}[Allocating social power to select individuals improves wisdom]\label{T3}
For social influence process~\eqref{eq1} with social power allocation $x\in\Delta_{n}$ and individuals' variances $\sigma^{2}\in\mathbb{R}^{n}$, let $\tau$ be a permutation of $\tau^{0}=(1,\dots,n)$. If there exist $m<n$ such that 
\begin{equation}\label{e3}
\sigma_{\tau_{j}}^{2}<\frac{2j-1}{n^2}\sum_{i=1}^{n}\sigma_{i}^{2} \quad \text{for all} \quad j\in\until{m},
\end{equation}
then system~\eqref{eq1} with any social power allocation $x$ satisfying $x_{\tau}\in\mathcal{\hat{A}}_{m}=\setdef{z\in \Delta^{\tau^{0}}_{n}}{z_{m+1}=0}$ asymptotically improves the wisdom.
\end{theorem}

\begin{proof}
Let $P$ be the permutation matrix associated with $\tau$, then $x_{\tau}=Px$, $\sigma_{\tau}^{2}=P\sigma^{2}$, and $x\in\mathcal{A}_{\sigma^{2}}$ if and only if $x_{\tau}\in\mathcal{A}_{\sigma_{\tau}^{2}}$.
Next, we prove that 
\begin{equation}\label{e4}
\frac{1}{j^2}\sum_{r=1}^{j}\sigma_{\tau_{r}}^{2}<\frac{1}{n^2}\sum_{i=1}^{n}\sigma_{i}^{2} \quad \text{for all} \quad j\in\until{m}.
\end{equation}
Using induction, for $j=1$, by~\eqref{e3} we have $\sigma_{\tau_{1}}^{2}<\frac{1}{n^2}\sum_{i=1}^{n}\sigma_{i}^{2}$. Suppose that 
\begin{align*}
\frac{1}{(h-1)^2}\sum_{r=1}^{h-1}\sigma_{\tau_{r}}^{2}<\frac{1}{n^2}\sum_{i=1}^{n}\sigma_{i}^{2} \quad \text{for all} \quad h-1\in\until{m-1}.
\end{align*}
Then, with $\sigma_{\tau_{h}}^{2}<\frac{2h-1}{n^2}\sum_{i=1}^{n}\sigma_{i}^{2}$ we obtain
\begin{align*}
\begin{split}
\frac{1}{h^2}\sum_{r=1}^{h}\sigma_{\tau_{r}}^{2}=&\frac{1}{h^2}\sigma_{\tau_{h}}^{2}+\frac{1}{h^2}\sum_{r=1}^{h-1}\sigma_{\tau_{r}}^{2}
<\frac{1}{h^2}\sigma_{\tau_{h}}^{2}+\frac{(h-1)^2}{h^{2}n^{2}}\sum_{i=1}^{n}\sigma_{i}^{2}
<\frac{1}{n^2}\sum_{i=1}^{n}\sigma_{i}^{2}.
\end{split}
\end{align*}
Therefore, \eqref{e4} holds. By \cref{L4}~\ref{L4-2}, we have $x_{\tau}\in\mathcal{A}_{\sigma_{\tau}^{2}}$, i.e., $x\in\mathcal{A}_{\sigma^{2}}$, which means social influence improves the wisdom asymptotically.
\end{proof}

Note that $\frac{2j-1}{n^2}\sum_{i=1}^{n}\sigma_{i}^{2}=\frac{2j-1}{n}\Ave(\sigma^{2})$. \Cref{T3} shows that selecting a subgroup of individuals whose variances are upper bounded, respectively, by specific fractions of individuals' average variance, and allocating social power to the select individuals according to the ordering of the upper bounds improves the wisdom. This is similar to the select-crowd strategy proposed in~\cite{AEM-JBS-RPL:14}. Moreover, note that $\mathcal{\hat{A}}_{m-1}\subseteq\mathcal{\hat{A}}_{m}$, which implies that if~\eqref{e3} holds for $m$, the select-crowd strategy works for all subgroup size no larger than $m$. Furthermore, the condition in \cref{T3} does not necessarily imply ordering-consistency. For example, consider social power allocation $x=(x_{1}, 0, 0, x_{4})^{\top}$ where $x_{1}\geq x_{4}>0$ and $x_{1}+ x_{4}=1$, then $\sigma^{2}=(3,24,20,2)^{\top}$
satisfies condition~\eqref{e3} with $m=2$ and $\tau=(1, 4, 2, 3)$. 

Motivated by \cref{L4}, we define the following hierarchy of individuals. 
\begin{definition}[$\tau$-hierarchy of individuals]
A group of individuals with variances $\sigma^{2}_{1},\dots,\sigma^{2}_{n}$ is said to admit a $\tau$-hierarchy if there exists a permutation $
\tau$ of $(1,\dots,n)$ such that
\begin{equation}\label{e3.1}
\frac{1}{i^2}\sum_{r=1}^{i}\sigma_{\tau_{r}}^{2}<\frac{1}{n^2}\sum_{j=1}^{n}\sigma_{j}^{2} \quad \text{for all} \quad i\in\until{n-1}.
\end{equation}
\end{definition}

Note that $\frac{1}{i^2}\sum_{r=1}^{i}\sigma_{r}^{2}$ is the variance of the initial collective estimate of the first $i$ individuals. Therefore, a group of individuals admit a $\tau$-hierarchy means that there exists a permutation $\tau$ such that after the permutation, the initial collective estimate of any top $m$ $(m<n)$ individuals outperforms the initial collective estimate of the entire group. Especially, by \cref{L4} and \cref{T3}, this hierarchical structure implies that after a permutation, selecting any top $m$ individuals and allocating social power decreasingly to them improves the wisdom. Based on the hierarchical structure, the next theorem proposes conditions with which we are able to determine whether social influence improves the wisdom by only checking the ordering of social power allocations. Without loss of generality, individuals are labelled in ascending order according to their variances. For permutation $\tau=(\tau_{1},\dots,\tau_{n})$, we say $(\tau_{p},\tau_{q})$ is an inversion if $\tau_{p}>\tau_{q}$ and $p<q$. For permutations $\tau$ and $\tau^{\prime}$, we write $\tau^{\prime}\prec\tau$ if $\tau^{\prime}$ can be obtained from $\tau$ by swapping inversions one by one, and $\tau^{\prime}\preceq\tau$ if $\tau^{\prime}\prec\tau$ or $\tau^{\prime}=\tau$. 

\begin{algorithm}
 \caption{\textbf{Modified permutation generation (MPG) algorithm}}
 \label{Alg1}
\begin{algorithmic}[1]
\item[\algorithmicrequire] $\sigma^{2}$ with $\sigma_{i}^{2}\leq\sigma_{i+1}^{2}$ for all $i\in\until{n-1}$
\item[\algorithmicensure] $\mathcal{T}_{\sigma^{2}}$ such that $\Delta^{\tau}_{n}\subset \mathcal{A}_{\sigma^{2}}$ for all $\tau\in\mathcal{T}_{\sigma^{2}}$ and $\Delta^{\tau}_{n}\not\subset \mathcal{A}_{\sigma^{2}}$ for all $\tau\in\mathcal{T}\setminus\mathcal{T}_{\sigma^{2}}$ where $\mathcal{T}$ is the set of all permutations of $(1,\dots,n)$
\STATE initialize $\mathcal{T}_{\sigma^{2}}:=\emptyset$; $s_{0}:=0$, $u_{0}:=\sigma^{2}_{1}$; $s_{i}:=\sum_{r=1}^{i}\sigma_{r}^{2}$ and $u_{i}:=\frac{i^2}{n^2}s_{n}-s_{i-1}$ for all $i\in\until{n}$
\STATE find $f_{i}=\max\setdef{j\in\{0,\dots,i-1\}}{\sigma^{2}_{i}\geq u_{j}}$ for all $i\in\until{n}$
\STATE 
$j:=0$
\FOR{$i$ from $1$ to $n$}
 \IF{$f_{i}=i-1$}
\STATE
$j:=j+1$\\
$h_{j}:=i$
\ENDIF 
\ENDFOR
\STATE 
$d:=j$, $h_{d+1}:=n+1$
\FOR{$i$ from $1$ to $d$}
\STATE generate all permutations of segment $(h_{i},\dots,h_{i+1}-1)$ using permutation generation algorithms (e.g., Heap's algorithm provided in~\cite{RS:77}) 
\ENDFOR
\STATE concatenate permutations of all segments, obtain $N=\prod_{i=1}^{d}(h_{i+1}-h_{i})!$ different permutations of $(1,\dots,n)$, denoted by $\tau^{1},\dots,\tau^{N}$
\FOR{$i$ from $1$ to $N$}
\IF{$\sigma^{2}_{\tau^{i}}$ satisfies~\cref{e3.1}}
\STATE$\mathcal{T}_{\sigma^{2}}:=\mathcal{T}_{\sigma^{2}}\cup\{\tau^{i}\}$
\ENDIF
\ENDFOR
\RETURN $\mathcal{T}_{\sigma^{2}}$
\end{algorithmic}
\end{algorithm}

\begin{theorem}[Improving the wisdom of hierarchical individuals]\label{T4}
For social influence process~\eqref{eq1} with social power allocation $x\in\Delta_{n}$ and individuals' variances $\sigma_{i}^{2}\leq\sigma_{i+1}^{2}$ for all $i\in\until{n-1}$, suppose that $\mathcal{T}$ is the set of all permutations of $\tau^{0}=(1,\dots,n)$ and $\Delta^{\tau}_{n}=\setdef{z\in\Delta_{n}\setminus\{ \frac{\mathbf{1}_{n}}{n}\}}{z_{\tau_{i}}\geq z_{\tau_{i+1}}, i\in\until{n-1}}$. Then,
\begin{enumerate}
\item ordering-consistency is sufficient for system~\eqref{eq1} to asymptotically improve the wisdom if and only if individuals admit the $\tau^{0}$-hierarchy; \label{T4-1}
\item if individuals admit a $\tau$-hierarchy, system~\eqref{eq1} with all social power allocations $x\in\bigcup_{\tau^{\prime}\in\setdef{\bar{\tau}\in\mathcal{T}}{\bar{\tau}\preceq\tau}}\Delta^{\tau^{\prime}}_{n}$ asymptotically improves the wisdom.\label{T4-2}
\item For~\cref{Alg1} and its output $\mathcal{T}_{\sigma^2}$,\label{T4-3}
\begin{enumerate}
\item system~\eqref{eq1} with all social power allocations $x\in\bigcup_{\tau\in\mathcal{T}_{\sigma^{2}}}\Delta^{\tau}_{n}$ asymptotically improves the wisdom;\label{T4-3.1}
\item for all $\tau\in\mathcal{T}\setminus\mathcal{T}_{\sigma^{2}}$, there exists $x\in\Delta^{\tau}_{n}$ such that system~\eqref{eq1} with social power allocation $x$ does not asymptotically improve the wisdom;\label{T4-3.2}
\item the running time of \cref{Alg1} is between $\mathcal{O}(n^2)$ and the running time of the enumeration method, which is at least $\mathcal{O}(n\times n!)$.\label{T4-3.3}
\end{enumerate}
\end{enumerate} 
\end{theorem}

\begin{proof}
Regarding~\ref{T4-1}, ordering-consistency is sufficient to improve the wisdom means that system~\eqref{eq1} improves the wisdom if it is ordering-consistent. Note that system~\eqref{eq1} is ordering-consistent if and only if $x\in\Delta^{\tau^{0}}_{n}$. By \cref{L4}~\ref{L4-1}, $\Delta^{\tau^{0}}_{n}\subset \mathcal{A}_{\sigma^{2}}$ if and only if~\eqref{eq2} holds, i.e., individuals admit the $\tau^{0}$-hierarchy. 

Regarding~\ref{T4-2}, let $P$ be the permutation matrix corresponding to $\tau$, then $Px\in\Delta_{n}^{\tau^{0}}$ and $P\sigma^{2}=\sigma^{2}_{\tau}$. Since~\eqref{e3.1} holds, \cref{L4}~\ref{L4-1} suggests that $\Delta_{n}^{\tau^{0}}\subset\mathcal{A}_{\sigma^{2}_{\tau}}$, which is equivalent to $\Delta_{n}^{\tau}\subset\mathcal{A}_{\sigma^{2}}$. Assume $\tau_{p}>\tau_{q}$ for $p<q$, let $\tilde{\tau}_{p}=\tau_{q}$, $\tilde{\tau}_{q}=\tau_{p}$, and $\tilde{\tau}_{l}=\tau_{l}$ for $l\in\until{n}\setminus\{p,q\}$, then $\tilde{\tau}\prec\tau$. Since $\tau_{l}=\tilde{\tau}_{l}$ for all $l\neq p,q$, we have $\sum_{r=1}^{i}\sigma_{\tilde{\tau}_{r}}^{2}=\sum_{r=1}^{i}\sigma_{\tau_{r}}^{2}$ for all $i\in\until{p-1}\cup\{q,\dots,n-1\}$. For $i\in\{p,\dots,q-1\}$, 
\begin{align*}
\sum_{r=1}^{i}\sigma_{\tilde{\tau}_{r}}^{2}=\sum_{r=1}^{p-1}\sigma_{\tau_{r}}^{2}+\sigma_{\tilde{\tau}_{p}}^{2}+\sum_{r=p+1}^{i}\sigma_{\tau_{r}}^{2}\leq\sum_{r=1}^{i}\sigma_{\tau_{r}}^{2}<\frac{i^2}{n^2}\sum_{j=1}^{n}\sigma_{j}^{2}
\end{align*}
due to $\sigma_{\tilde{\tau}_{p}}^{2}=\sigma_{\tau_{q}}^{2}\leq\sigma_{\tau_{p}}^{2}$. That is, \eqref{e3.1} holds for $\tilde{\tau}$, which implies $\Delta_{n}^{\tilde{\tau}}\subset\mathcal{A}_{\sigma^{2}}$ and consequently $\Delta_{n}^{\tilde{\tau}}\subset\mathcal{A}_{\sigma^{2}}$ for all $\tilde{\tau}\preceq\tau$. In conclusion, system~\eqref{eq1} asymptotically improves the wisdom with all social power allocations $x\in\bigcup_{\tau^{\prime}\in\setdef{\bar{\tau}\in\mathcal{T}}{\bar{\tau}\preceq\tau}}\Delta^{\tau^{\prime}}_{n}$. 

Regarding~\ref{T4-3.1}, by steps $11,12,13$ of \cref{Alg1}, $\sigma^{2}_{\tau}$ satisfies~\eqref{e3.1} for any $\tau\in\mathcal{T}_{\sigma^{2}}$. Statement~\ref{T4-2} suggests system~\eqref{eq1} asymptotically improves the wisdom with all social power allocations $x\in\Delta_{n}^{\tau}$ for all $\tau\in\mathcal{T}_{\sigma^{2}}$. Regarding~\ref{T4-3.2}, suppose $\tau\in\mathcal{T}\setminus\mathcal{T}_{\sigma^{2}}$. There are two situations. One is that $\tau$ is obtained from step $10$, i.e., $\tau\in\{\tau^{1},\dots,\tau^{N}\}$, but is not added to $\mathcal{T}_{\sigma^{2}}$ in the following steps; the other is that $\tau\notin\{\tau^{1},\dots,\tau^{N}\}$. The first situation means that $\sigma^{2}_{\tau}$ does not satisfy~\eqref{e3.1}. For the second situation, note that from steps $3$ to $7$ the algorithm divides $(1,\dots,n)$ into $d$ segments, that is, $(h_{i},\dots,h_{i+1}-1)$ for $i\in\until{d}$, and $\{\tau^{1},\dots,\tau^{N}\}$ contains all the permutations of $(1,\dots,n)$ obtained by concatenating permutations of $d$ segments. Therefore, $\tau\notin\{\tau^{1},\dots,\tau^{N}\}$ implies that there exists $j^{*}\in\until{n}$ and $i^{*}\in\until{d}$ such that $j^{*}\in\{h_{i^{*}},\dots,h_{i^{*}+1}-1\}$ and $\tau_{j^{*}}\notin\{h_{i^{*}},\dots,h_{i^{*}+1}-1\}$. Without loss of generality, let $n\geq\tau_{j^{*}}\geq h_{i^{*}+1}>j^{*}$. Since $f_{h_{i^{*}+1}}=h_{i^{*}+1}-1$, we have $\sigma^{2}_{h_{i^{*}+1}}\geq u_{h_{i^{*}+1}-1}$, i.e., $s_{h_{i^{*}+1}-2}+\sigma^{2}_{h_{i^{*}+1}}\geq \frac{(h_{i^{*}+1}-1)^2}{n^2}s_{n}$.
Hence,
\begin{align*}
\sum_{r=1}^{h_{i^{*}+1}-1}\sigma^{2}_{\tau_{r}}=\sum_{\substack{r=1 \\ r\neq j^{*}}}^{h_{i^{*}+1}-1}\sigma^{2}_{\tau_{r}}+\sigma^{2}_{\tau_{j^{*}}}\geq s_{h_{i^{*}+1}-2}+\sigma^{2}_{h_{i^{*}+1}}\geq \frac{(h_{i^{*}+1}-1)^2}{n^2}s_{n},
\end{align*}
where the second last inequality is implied by $\sigma_{i}^{2}\leq\sigma_{i+1}^{2}$ for all $i\in\until{n-1}$. Therefore, $\sigma^{2}_{\tau}$ does not satisfy~\eqref{e3.1} for $i=h_{i^{*}+1}-1$. In conclusion, $\sigma^{2}_{\tau}$ does not satisfy~\eqref{e3.1} for any $\tau\in\mathcal{T}\setminus\mathcal{T}_{\sigma^{2}}$. By \cref{L4}~\ref{L4-1} and the proof of statement~\ref{T4-2}, $\Delta_{n}^{\tau}\not\subset\mathcal{A}_{\sigma^{2}}$ for all $\tau\in\mathcal{T}\setminus\mathcal{T}_{\sigma^{2}}$, which completes the proof. Regarding~\ref{T4-3.3}, by~\cite{RS:77}, the running time of permutation generation algorithm employed in step $9$ is at least $\mathcal{O}(m!)$ to generate all permutations of $m$ numbers. The time complexity of step $2$ is $\mathcal{O}(n^2)$. Steps $3$ to $7$ need to execute elementary operations $3n$ times. Steps $8$ and $9$ need to call the permutation generation algorithm $d$ times with total running time at least $\sum_{i=1}^{d}(h_{i+1}-h_{i})!$. Step $10$ needs to concatenate $\sum_{i=1}^{d}(h_{i+1}-h_{i})!$ segments and obtains $N=\prod_{i=1}^{d}(h_{i+1}-h_{i})!$ permutations of $(1,\dots,n)$. Hence, the running time of step $10$ is $N$. Moreover, the complexity of steps $11$ to $13$ is $\mathcal{O}(nN)$. In the case that $d=n$, i.e., $f_{i}=i-1$ for all $i\in\until{n}$, we have $\sum_{i=1}^{d}(h_{i+1}-h_{i})!=n$ and $N=1$. Thus, the running time of \cref{Alg1} is $\mathcal{O}(n^2)$. In the case that $d=1$, i.e., $f_{i}=i-1$ only for $i=1$, we have $\sum_{i=1}^{d}(h_{i+1}-h_{i})!=N=n!$ and \cref{Alg1} generates all $n!$ permutations of $(\sigma^{2}_{1},\dots,\sigma^{2}_{n})$ and runs steps $12$ and $13$ for each permutation, which is exactly the enumeration method with running time at least $\mathcal{O}(n\times n!)$. In conclusion, the running time of \cref{Alg1} is between $\mathcal{O}(n^2)$ and the running time of the enumeration method.
\end{proof}

\begin{remark}
The existence of any permutation of $\sigma^{2}$ satisfying~\eqref{e3.1} is equivalent to the existence of any solution of the following $0$-$1$ integer programming without optimization:
\begin{equation}\label{e3.3}
\begin{aligned}
 &z\in\mathbb{R}^{n^2}\\
\textup{s.t.} \ \ \ &Az<b\\
&A^{\prime}z=\mathbf{1}_{2n},\\
&z_{i}\in\{0,1\} \ \textup{for all} \ i\in\until{n^{2}},
\end{aligned}
\end{equation}
where $z=[P_{1}  \  \dots \ P_{n}]^{\top}$ with $P_{i}\in\mathbb{R}^{1\times n}$ being the $i$-th row of the permutation matrix corresponding to $\tau$, $A=Q\otimes(\sigma^{2})^{\top}\in\mathbb{R}^{n-1\times n^{2}}$ with $Q\in\mathbb{R}^{n-1\times n}$, $Q_{ij}=1$ for $i\geq j$ and $Q_{ij}=0$ otherwise, $b\in\mathbb{R}^{n-1}$ with $b_{i}=\frac{i^2}{n^2}\sum_{r=1}^{n}\sigma_{r}^{2}$, $A^{\prime}=[I_{n}\otimes \mathbf{1}_{n} \ \ \mathbf{1}_{n}\otimes I_{n}]^{\top}\in\mathbb{R}^{2n\times n^{2}}$. Karp~\cite{RMK:72} proved that problem~\eqref{e3.3} is NP-complete. Equivalently, \cref{Alg1} finds all the feasible solutions of problem~\eqref{e3.3}. There may be other algorithms which can solve problem~\eqref{e3.3}, for example, the intlinprog function of MATLAB. However, the intlinprog function only provides one feasible solution. 
\end{remark}

\Cref{T4} provides several results based on the idea of partitioning the $n$-simplex into hypertriangles such that social power allocations in different hypertriangles are ordered differently. \Cref{T4}~\ref{T4-1} shows that ordering-consistency is sufficient for improving the wisdom if and only if individuals admit a $\tau^{0}$-hierarchy. \Cref{T4}~\ref{T4-2} suggests that if individuals with ascending variances admit a $\tau$-hierarchy with $\tau\neq\tau^{0}$, social influence with all social power allocations in certain orderings improves the wisdom; and there exists more than one such ordering. To address this issue, we propose \cref{Alg1}, which, for given individuals' variances, outputs all the orderings of social power allocations that improve the wisdom. In \cref{f2-c,f2-d}, triangles $DCE$ and $DBE$ (excluding the points $D$) are, respectively, the ordering-consistent regions for the given variances. They are both contained inside the improvement regions since individuals with $\sigma^{2}=(1, 4, 9)^{\top}$ and $\sigma^{2}=(1, 2, 16)^{\top}$ both admit the $\tau^{0}$-hierarchy. Moreover, in \cref{f2-d}, triangle $DCE$ (excluding the point $D$) is also contained in the improvement region since individuals with $\sigma^{2}=(1, 2, 16)^{\top}$ also admit the $(2,1,3)$-hierarchy. 

\section{The role of social power in undermining the wisdom}\label{S4}
In this section, we investigate the question of when social influence asymptotically undermines the wisdom of crowds. For a permutation $\tau$ of $(1,\dots,n)$, denote by $\tau^{-1}=(\tau_{n},\dots,\tau_{1})$ its inverse. Given $\sigma^{2}\in\mathbb{R}^{n}$, let $\mathcal{\bar{A}}_{\sigma^{2}}=\setdef{z\in\Delta_{n}}{\mathcal{E}_{\sigma^{2}}(z)>\mathcal{E}_{\sigma^{2}}(\frac{\mathbf{1}_{n}}{n})}$ be the undermining region. The following theorem shows that if the ordering of the influence system's social power allocations is in reverse to the ordering of individuals' accuracy, social influence undermines the wisdom. 
\begin{theorem}[Ordering-reverse influence system undermines the wisdom]\label{T5}
For social influence process~\eqref{eq1} with social power allocation $x\in\Delta_{n}$ and individuals' variances $\sigma^{2}\in\mathbb{R}^{n}$, let $\tau$ be a permutation such that $\sigma_{\tau_{i}}^{2}\leq\sigma_{\tau_{i+1}}^{2}$ for all $i\in\until{n-1}$, $\Delta^{\tau}_{n}=\setdef{z\in\Delta_{n}\setminus\{ \frac{\mathbf{1}_{n}}{n}\}}{z_{\tau_{i}}\geq z_{\tau_{i+1}}, i\in\until{n-1}}$. System~\eqref{eq1} with all social power allocations $x\in\Delta_{n}^{\tau^{-1}}$ asymptotically undermines the wisdom.
\end{theorem}
\begin{proof}
Since $x\in\Delta_{n}^{\tau^{-1}}$, we obtain $x_{\tau_{i}}\leq x_{\tau_{i+1}}$ for all $i\in\until{n-1}$. Without loss of generality, assume that $x_{\tau_{m}}<\frac{1}{n}$ and $x_{\tau_{m+1}}\geq\frac{1}{n}$ with $m<n$. Then, we have 
\begin{align*}
&\mathcal{E}_{\sigma_{\tau}^{2}}(x_{\tau})-\mathcal{E}_{\sigma_{\tau}^{2}}(\frac{\mathbf{1}_{n}}{n})
=\sum_{j=1}^{n}x_{\tau_{j}}^{2}\sigma_{\tau_{j}}^{2}-\frac{1}{n^2}\sum_{j=1}^{n}\sigma_{\tau_{j}}^{2}\\
=&\sum_{j=1}^{m}(x_{\tau_{j}}+\frac{1}{n})(x_{\tau_{j}}-\frac{1}{n})\sigma_{\tau_{j}}^{2}
+\sum_{j=m+1}^{n}(x_{\tau_{j}}+\frac{1}{n})(x_{\tau_{j}}-\frac{1}{n})\sigma_{\tau_{j}}^{2}\\
>&\frac{2\sigma_{\tau_{m+1}}^{2}}{n}\sum_{j=m+1}^{n}(x_{\tau_{j}}-\frac{1}{n})-\frac{2\sigma_{\tau_{m}}^{2}}{n}\sum_{j=1}^{m}(\frac{1}{n}-x_{\tau_{j}})
\geq 0,
\end{align*} 
where the strictly inequality holds since $x\neq\frac{\mathbf{1}_{n}}{n}$. 
Therefore, we obtain $x_{\tau}\in\mathcal{\bar{A}}_{\sigma_{\tau}^{2}}$, which is equivalent to $x\in\mathcal{\bar{A}}_{\sigma^{2}}$. In conclusion, system~\eqref{eq1} with all social power allocations $x\in\Delta_{n}^{\tau^{-1}}$ asymptotically undermines the wisdom.
\end{proof}

\Cref{T5} suggests that if more accurate individuals are allocated less social power, wisdom of crowds is undermined. In other words, for any given $\sigma^{2}$, there exists an ordering of social power allocations such that social influence with social power allocations in that ordering undermines the wisdom, see the orange areas in \cref{f2-a,f2-b}. Comparing \cref{T5} with \cref{T3,T4}, to undermine the wisdom just needs to allocate more social power to less accurate individuals, while generally allocating more social power to more accurate individuals, i.e., ordering-consistency, is not sufficient to improve the wisdom. 

\begin{theorem}[Connection between improving and undermining of the wisdom]\label{T6}
For social influence process~\eqref{eq1} with social power allocation $x\in\Delta_{n}$ and individuals' variances $\sigma^{2}\in\mathbb{R}^{n}$, let $\tau$ be a permutation of $(1,\dots,n)$, $\Delta^{\tau}_{n}=\setdef{z\in\Delta_{n}\setminus\{ \frac{\mathbf{1}_{n}}{n}\}}{z_{\tau_{i}}\geq z_{\tau_{i+1}}, i\in\until{n-1}}$. If system~\eqref{eq1} with all social power allocations $x\in\Delta_{n}^{\tau}$ asymptotically improves the wisdom, then system~\eqref{eq1} with all social power allocations $x\in\Delta_{n}^{\tau^{-1}}$ asymptotically undermines the wisdom.
\end{theorem}
\begin{proof}
Since system~\eqref{eq1} with all social power allocations $x\in\Delta_{n}^{\tau}$ asymptotically improves the wisdom, we have $\Delta_{n}^{\tau}\subset\mathcal{A}_{\sigma^{2}}$, which, by \cref{L4}~\ref{L4-1}, suggests  
\begin{align*}
\frac{1}{i^2}\sum_{r=1}^{i}\sigma_{\tau_{r}}^{2}<\frac{1}{n^2}\sum_{j=1}^{n}\sigma_{j}^{2} 
\end{align*}
for all $i\in\until{n-1}$. Therefore, 
\begin{align*}
\sum_{j=1}^{n}\sigma_{j}^{2}=\sum_{r=1}^{i-1}\sigma_{\tau_{r}}^{2}+\sum_{r=i}^{n}\sigma_{\tau_{r}}^{2}<\frac{(i-1)^2}{n^2}\sum_{j=1}^{n}\sigma_{j}^{2}+\sum_{r=i}^{n}\sigma_{\tau_{r}}^{2},
\end{align*}
which implies 
\begin{equation}\label{e7.1}
\sum_{r=i}^{n}\sigma_{\tau_{r}}^{2}>(1-\frac{(i-1)^2}{n^2})\sum_{j=1}^{n}\sigma_{j}^{2} \quad \text{for all} \quad i\in\{2,\dots,n\}.
\end{equation}
For $x\in\Delta_{n}^{\tau^{-1}}$, i.e., $x_{\tau_{i}}\leq x_{\tau_{i+1}}$, $i\in\until{n-1}$, denote $q_{1}=x_{\tau_{1}}$ and $q_{i}=x_{\tau_{i}}-x_{\tau_{i-1}}$ for $i\in\{2,\dots,n\}$. Then, we have $q_{i}\geq 0$ and $x_{\tau_{i}}=\sum_{r=1}^{i}q_{r}$. Moreover, 
\begin{align*}
\sum_{i=1}^{n}x_{\tau_{i}}^{2}\sigma_{\tau_{i}}^{2}
=\sum_{i=1}^{n}(\sum_{r=1}^{i}q_{r})^{2}\sigma_{\tau_{i}}^{2}
=&\sum_{i=1}^{n}\sigma_{\tau_{i}}^{2}\sum_{r=1}^{i}q_{r}^{2}+2\sum_{i=2}^{n}\sigma_{\tau_{i}}^{2}\sum_{r=1}^{i}q_{r}\sum_{j=r+1}^{i}q_{j}\\
=&q_{1}^{2}\sum_{r=1}^{n}\sigma_{\tau_{r}}^{2}+\sum_{i=2}^{n}q_{i}^{2}\sum_{r=i}^{n}\sigma_{\tau_{r}}^{2}+2\sum_{i=2}^{n}q_{i}\sum_{j=1}^{i-1}q_{j}\sum_{r=i}^{n}\sigma_{\tau_{r}}^{2}.
\end{align*} 
By~\eqref{e7.1}, we obtain
\begin{align*}
\sum_{i=1}^{n}x_{\tau_{i}}^{2}\sigma_{\tau_{i}}^{2}>&(q_{1}^{2}+\sum_{i=2}^{n}(1-\frac{(i-1)^2}{n^2})(q_{i}^{2}+2q_{i}\sum_{j=1}^{i-1}q_{j}))\sum_{r=1}^{n}\sigma_{\tau_{r}}^{2},
\end{align*} 
where
\begin{align*}
q_{1}^{2}+\sum_{i=2}^{n}(1-\frac{(i-1)^2}{n^2})(q_{i}^{2}+2q_{i}\sum_{j=1}^{i-1}q_{j})=&(\sum_{i=1}^{n}q_{i})^{2}-\frac{1}{n^2}\sum_{i=2}^{n}(i-1)^{2}(q_{i}^{2}+2q_{i}\sum_{j=1}^{i-1}q_{j})\\
=&x_{\tau_{n}}^{2}\!-\!\frac{1}{n^2}\sum_{i=2}^{n}(i\!-\!1)^{2}(x_{\tau_{i}}^{2}\!-\!x_{\tau_{i-1}}^{2}\!)\!\\
=&\frac{1}{n^2}\sum_{i=1}^{n}(2i-1)x_{\tau_{i}}^{2}
\end{align*} 
due to $\sum_{i=2}^{n}(i\!-\!1)^{2}(x_{\tau_{i}}^{2}\!-\!x_{\tau_{i-1}}^{2}\!)=(n-1)^{2}x_{\tau_{n}}^{2}-\sum_{i=1}^{n-1}(2i-1)x_{\tau_{i}}^{2}$. Note that 
\begin{align*}
\sum_{i=1}^{n}(2i-1)x_{\tau_{i}}^{2}-1=&\sum_{i=1}^{n}(2i-1)x_{\tau_{i}}^{2}-(\sum_{i=1}^{n}x_{\tau_{i}})^{2}\\
=&2\sum_{i=2}^{n}(i-1)x_{\tau_{i}}^{2}-2\sum_{i=2}^{n}\sum_{j=1}^{i-1}x_{\tau_{i}}x_{\tau_{j}}\\
=&2\sum_{i=2}^{n}x_{\tau_{i}}((i-1)x_{\tau_{i}}-\sum_{j=1}^{i-1}x_{\tau_{j}})\geq 0
\end{align*}
because $x_{\tau_{i}}\geq x_{\tau_{j}}$ for all $j\in\until{i-1}$. Hence, we obtain 
\begin{align*}
\sum_{i=1}^{n}x_{\tau_{i}}^{2}\sigma_{\tau_{i}}^{2}>\frac{1}{n^2}\sum_{i=1}^{n}(2i-1)x_{\tau_{i}}^{2}\sum_{r=1}^{n}\sigma_{r}^{2}\geq\frac{1}{n^2}\sum_{r=1}^{n}\sigma_{r}^{2},
\end{align*} 
which implies that social influence asymptotically undermines the wisdom. 
\end{proof}

\Cref{T6} implies that if there exists a permutation $\tau$ of $(1,\dots,n)$ such that system~\eqref{eq1} with all social power allocations $x\in\Delta_{n}^{\tau}$ asymptotically improves the wisdom, then system~\eqref{eq1} with all social power allocations $x\in\Delta_{n}^{\tau^{-1}}$ asymptotically undermines the wisdom, see \cref{f2-c,f2-d}. Moreover, \cref{T4} suggests that there may be more than one such permutation, as shown in \cref{f2-d}, and \cref{Alg1} can find all these permutations for given $\sigma^{2}$. However, the converse of \cref{T6} does not necessarily hold, i.e., all social power allocations $x\in\Delta_{n}^{\tau}$ asymptotically undermine the wisdom does not imply all social power $x\in\Delta_{n}^{\tau^{-1}}$ asymptotically improves the wisdom generally, see \cref{f2-a,f2-b} for counter examples. A direct corollary combining \cref{T4} and \cref{T6} is as follows. 

\begin{corollary}
For social influence process~\eqref{eq1} with social power allocation $x\in\Delta_{n}$ and individuals' variances satisfying $\sigma_{i}^{2}\leq\sigma_{i+1}^{2}$ for all $i\in\until{n-1}$, let $\tau$ be a permutation of $(1,\dots,n)$, $\Delta^{\tau}_{n}=\setdef{z\in\Delta_{n}\setminus\{ \frac{\mathbf{1}_{n}}{n}\}}{z_{\tau_{i}}\geq z_{\tau_{i+1}}, i\in\until{n-1}}$. If individuals admit a $\tau$-hierarchy, then system~\eqref{eq1} with all social power allocations $x\in\bigcup_{\tau^{\prime}\in\setdef{\bar{\tau}\in\mathcal{T}}{\bar{\tau}\preceq\tau}}\Delta^{(\tau^{\prime})^{-1}}_{n}$ asymptotically undermines the wisdom.
\end{corollary}

\section{Improving and undermining the wisdom with the French-DeGroot opinion dynamics}\label{S5}

As reported in~\cite{JB-DB-DC:17} and~\cite{GM-GGdP:15}, one of the most significant empirical findings on wisdom of crowds is the \emph{PAP hypothesis}, which attributes the improvement of wisdom to that more accurate individuals are more resistant to social influence, while less accurate individuals are more susceptible to social influence. In this section, we apply our theoretical results to the FD influence process to examine the effects of the \emph{PAP hypothesis}. 
\subsection{The French-DeGroot opinion dynamics and the PAP hypothesis}

Suppose that individuals interact their estimates according to the FD model:
\begin{equation}\label{e1}
y(k+1)=Wy(k),
\end{equation}
where $W$ is the row-stochastic influence matrix. More specifically, individual $i$ updates its estimate to a convex combination of estimates of others and itself, that is, $y_{i}(k+1)=\sum_{j=1}^{n}W_{ij}y_{j}(k)$, where $W_{ij}$ is the influence weight individual $i$ assigns to individual $j$ and $W_{ii}$ is its self-weight. Sociologically, $W_{ii}$ is interpreted as 
individual $i$'s self-appraisal and indicates its resistance to social influence, while $1-W_{ii}$ represents its susceptibility to social influence. Let $\gamma_{i}=1-W_{ii}$, then there exists row-stochastic, zero-diagonal matrix $C\in\mathbb{R}^{n\times n}$, called the relative interaction matrix, such that $C_{ij}=0$ for all $i=j$ and $W_{ij}=\gamma_{i}C_{ij}$ otherwise. Denote by $\gamma\in\mathbb{R}^{n}$ the susceptibility vector,~\eqref{e1} can be written as
\begin{equation}\label{e2}
y(k+1)=[\gamma]Cy(k)+(I_{n}-[\gamma])y(k)
\end{equation}
with $W=[\gamma]C+I_{n}-[\gamma]$. Note that $\gamma_{i}=0$ means that individual $i$ does not take into account others' estimates, and $\gamma_{i}=1$ indicates that individual $i$ does not consider its own estimates. We assume that every individual is open-minded about others' estimates and there exists at least one individual who is confident in its own estimates, i.e., $\gamma\in\Gamma=\setdef{z\in\mathbb{R}^{n}\setminus\{\mathbf{1}_{n}\}}{1\geq z_{i}>0, i\in\until{n}}$. With a slight abuse of terminology, we regard system~\eqref{e2} as a population $\mathcal{P}_{\sigma^{2},\gamma}$ in a relative interaction network $\mathcal{G}(C)$, where a population is a group of individuals with some characteristics, which mainly concerned here are individual's susceptibility and variance. 
\begin{remark}[Model advantages]
It is reported in~\cite{JB-DB-DC:17} that over $80\%$ people in the experiments displayed behavior of opinion updating consistent with system~\eqref{e2}. Moreover,~\eqref{e2} can model not only influence networks in which the relative interaction matrix $C$ is determined or accessible by individuals, such as, face to face discussion or social media networks, but also the influence networks in which $C$ is unknown to individuals, such as the Delphi model. Recall that the four basic features of the Delphi method are anonymity, iteration, controlled feedback and statistical aggregation. If we assume that $C$ is controlled by a moderator and is unknown to individuals, then system~\eqref{e2} is a well-posed influence network formulation for the Delphi method, where anonymity and controlled feedback are guaranteed by the assumption that $C_{ij}$ is unknown to individuals; iteration is achieved by the dependence of $y(k)$ upon time scale $k$; and the statistical aggregation is arithmetic average.
\end{remark}

Motivated by the \emph{PAP hypothesis}, we have the following definitions for populations, which are presented analogously to \cref{D2}.
\begin{definition}[PAP populations]\label{D3}
For a population $\mathcal{P}_{\sigma^{2},\gamma}$ consisting of individuals $\until{n}$ with variances $\sigma^{2}\in\mathbb{R}^{n}$ and susceptibilities $\gamma\in\setdef{z\in\mathbb{R}^{n}\setminus\{\mathbf{1}_{n}\}}{1\geq z_{i}>0, i\in\until{n}}$, we say $\mathcal{P}_{\sigma^{2},\gamma}$ is a
\begin{enumerate}
\item PAP population if $\gamma_{i}\geq\gamma_{j}$ for all $\sigma^{2}_{i}\geq\sigma^{2}_{j}$, and $\gamma\in\Span{\mathbf{1}_{n}}$ only if $\sigma^{2}\in\Span{\mathbf{1}_{n}}$;

\item strong PAP population if $\mathcal{P}_{\sigma^{2},\gamma}$ is a PAP population, and $\gamma_{j}\sigma^{2}_{i}\geq\gamma_{i}\sigma^{2}_{j}$ for all $\sigma^{2}_{i}\geq\sigma^{2}_{j}$; 

\item maximal PAP population if $\frac{\gamma_{i}}{\gamma_{j}}=\frac{\sigma^{2}_{i}}{\sigma^{2}_{j}}$ for all $i,j\in\until{n}$.
\end{enumerate} 
\end{definition}

In a PAP population, less accurate individuals are more susceptible to social influence, while more accurate individuals are more resistant to social influence. Strong PAP population requires that more accurate individuals are more resistant to social influence compared with less accurate individuals, but can not be too much more resistant. Maximal PAP population requires that the proportions of individuals' susceptibilities are exactly the same as the proportions of their variances. 
\subsection{Improving and optimizing the wisdom in irreducible influence networks}
It is well known that for system~\eqref{e2}, $\lim_{k\to\infty}y(k)$ exists if and only if every sink SCC of $\mathcal{G}(W)$ is aperiodic; if $\mathcal{G}(W)$ is strongly connected and aperiodic (or equivalently, $W$ is primitive), $y(k)$ reaches consensus asymptotically, that is, $\lim_{k\to\infty}y(k)=\omega^{\top}y(0)\mathbf{1}_{n}$, where $\omega\in\interior{\Delta_{n}}$ is the left dominant eigenvector of $W$ satisfying $\omega^{\top}W=\omega^{\top}$~\cite{FB:22}.
If the relative interaction matrix $C$ is reducible, then $W$ is reducible since $\gamma>0$. Consequently, individuals in sink SCCs have positive social power and all others' social power is $0$. In this case, whether system~\eqref{e2} improves the wisdom or not is decided by the social power of individuals in sink SCCs, which corresponds to row-stochastic and irreducible submatrices of $C$. For simplicity, we assume $C$ is irreducible, and $c\in\Delta_{n}$ is its left dominant eigenvector. Since $\gamma\in\Gamma=\setdef{z\in\mathbb{R}^{n}\setminus\{\mathbf{1}_{n}\}}{1\geq z_{i}>0, i\in\until{n}}$, $W=[\gamma]C+I_{n}-[\gamma]$ is irreducible and aperiodic~\cite[Corollary 8.4.7]{RAH-CRJ:12}. Therefore, we have $x^{\top}W=x^{\top}[\gamma]C+x^{\top}(I_{n}-[\gamma])=x^{\top}$,
by which there holds $x^{\top}[\gamma]C=x^{\top}[\gamma]$. That is, $x^{\top}[\gamma]>0$ is a left eigenvector of $C$ associated with eigenvalue $1$. By Perron-Frobenius theorem~\cite[Theorem 2.12]{FB:22}, $[\gamma]x=\eta c$ with $\eta=\sum_{i=1}^{n}\gamma_{i}x_{i}$, which implies
\begin{equation}\label{e8}
\frac{x_{i}}{x_{j}}=\frac{c_{i}}{c_{j}}\frac{\gamma_{j}}{\gamma_{i}} \quad \text{for all} \quad i,j\in\until{n}.
\end{equation}

\begin{corollary}[Improving/optimizing the wisdom with PAP populations in irreducible influence networks]\label{C1}
For system~\eqref{e2} with population $\mathcal{P}_{\sigma^{2},\gamma}$ and relative interaction network $\mathcal{G}(C)$, where $\gamma\in\Gamma$ with $\Gamma=\setdef{z\in\mathbb{R}^{n}\setminus\{\mathbf{1}_{n}\}}{1\geq z_{i}>0, i\in\until{n}}$, $C$ is irreducible with dominant left eigenvector $c\in\interior\Delta_{n}$, the following statements hold:
\begin{enumerate}
\item assume that $\sigma^{2}\notin\Span{\mathbf{1}_{n}}$, then system~\eqref{e2} asymptotically improves the wisdom if $\frac{\sigma_{j}^{2}}{\sigma_{i}^{2}}\frac{\gamma_{i}}{\gamma_{j}}\leq\frac{c_{i}}{c_{j}}\leq\frac{\gamma_{i}}{\gamma_{j}}$ for all $\sigma_{i}^{2}\geq\sigma_{j}^{2}$ and $\frac{c_{i}}{c_{j}}<\frac{\gamma_{i}}{\gamma_{j}}$ for at least one $\sigma_{i}^{2}\geq\sigma_{j}^{2}$;\label{C1-1}

\item assume that $\mathcal{P}_{\sigma^{2},\gamma}$ is a maximal PAP population, then system~\eqref{e2} asymptotically optimizes the wisdom if and only if $\mathcal{G}(C)$ is democratic; \label{C1-2}

\item assume that individuals admit a $\tau$-hierarchy, then system~\eqref{e2} asymptotically improves the wisdom if $\frac{c_{\tau_{i}}}{c_{\tau_{i+1}}}\geq\frac{\gamma_{\tau_{i}}}{\gamma_{\tau_{i+1}}}$ for all $i\in\until{n-1}$ with at least one strict inequality holding. \label{C1-3}

\end{enumerate}
\end{corollary}

\Cref{C1}~\ref{C1-1} requires that $\frac{c_{i}}{c_{j}}\leq\frac{\gamma_{i}}{\gamma_{j}}$ for $\sigma_{i}^{2}\geq\sigma_{j}^{2}$, which holds if $\gamma_{i}>\gamma_{j}$ and $c_{i}<c_{j}$. That is, more accurate individual is more resistant to social influence and has larger network centrality. If the condition of \cref{C1}~\ref{C1-1} holds, then system~\eqref{e2} is gap-consistent, thus improves the wisdom. However, $\mathcal{P}_{\sigma^{2},\gamma}$ is not necessarily a PAP population under the condition of \cref{C1}~\ref{C1-1}. Generally, the \emph{PAP hypothesis} is not adequate to improve the wisdom in the FD model. \Cref{C1} shows that even a maximal PAP population is not sufficient to improve the wisdom. In \cref{f3-d,f3-e}, $\mathcal{E}^{1}$ and $\mathcal{E}^{4}$ are the trajectories of the variances of collective estimates of system~\eqref{e2} under relative influence network $\mathcal{G}(C^2)$ depicted in \cref{f3-b} with centralities scores $c=(0.4, 0.4, 0.2)^{\top}$. Individuals' susceptibilities are $\gamma=(0.1, 0.15, 0.1)^{\top}$ and $\gamma=(0.1, 0.8, 0.6)^{\top}$, which satisfy the conditions in \cref{C1}~\ref{C1-1} and \cref{C1}~\ref{C1-3}, respectively. Thus, system~\eqref{e2} improves the wisdom. 

\begin{figure}[t]
\setlength{\belowcaptionskip}{-1cm}
\centering
 \subfloat[t][ Democratic network $\mathcal{G}(C^{1})$]{
\begin{minipage}[t]{0.29\textwidth}\label{f3-a}
 \centering
  \includegraphics[width=\hsize]{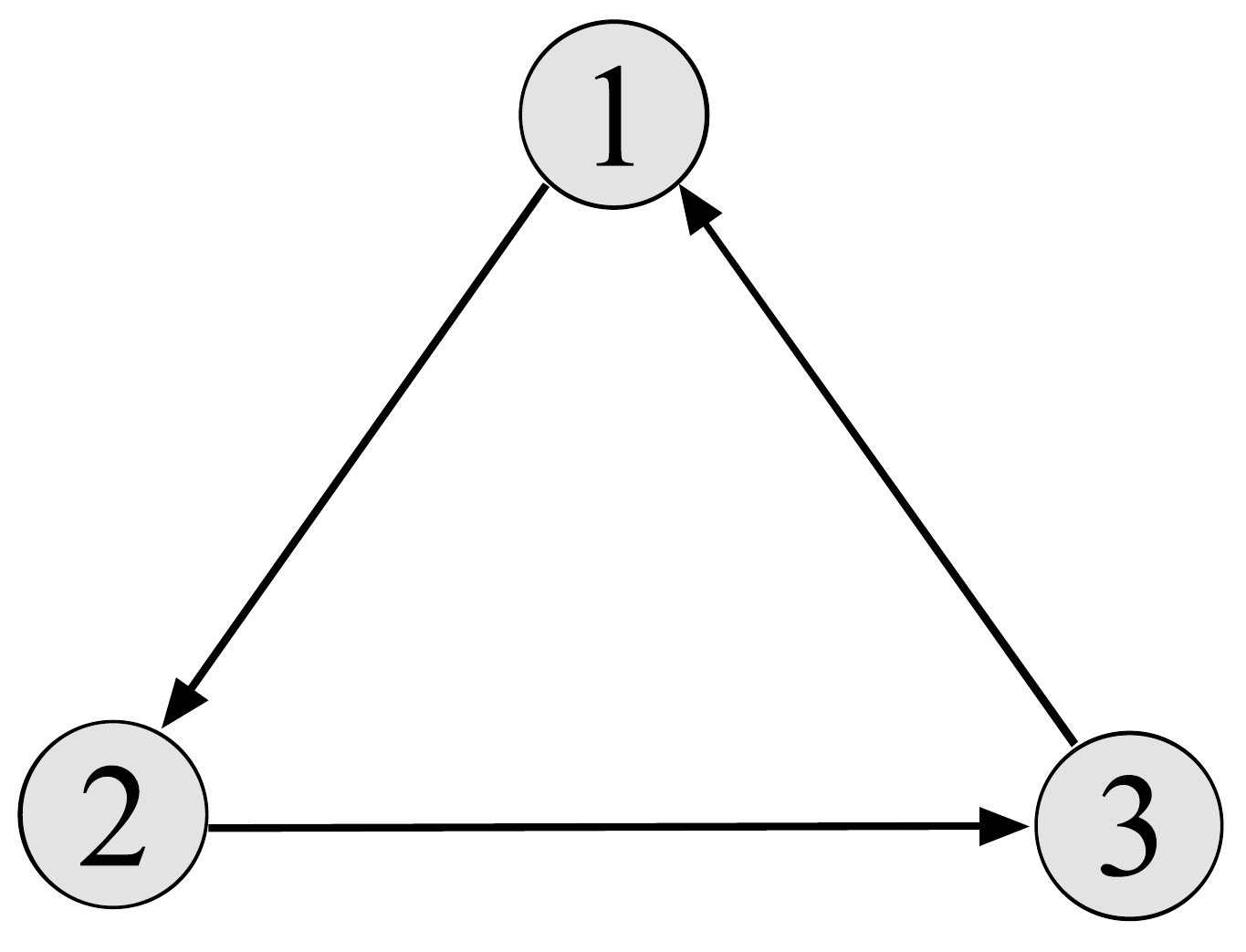}\\
 \end{minipage}
}
\hspace{1ex}
\subfloat[t][Irreducible network $\mathcal{G}(C^{2})$]{
\begin{minipage}[t]{0.3\textwidth}\label{f3-b}
 \centering
  \includegraphics[width=\hsize]{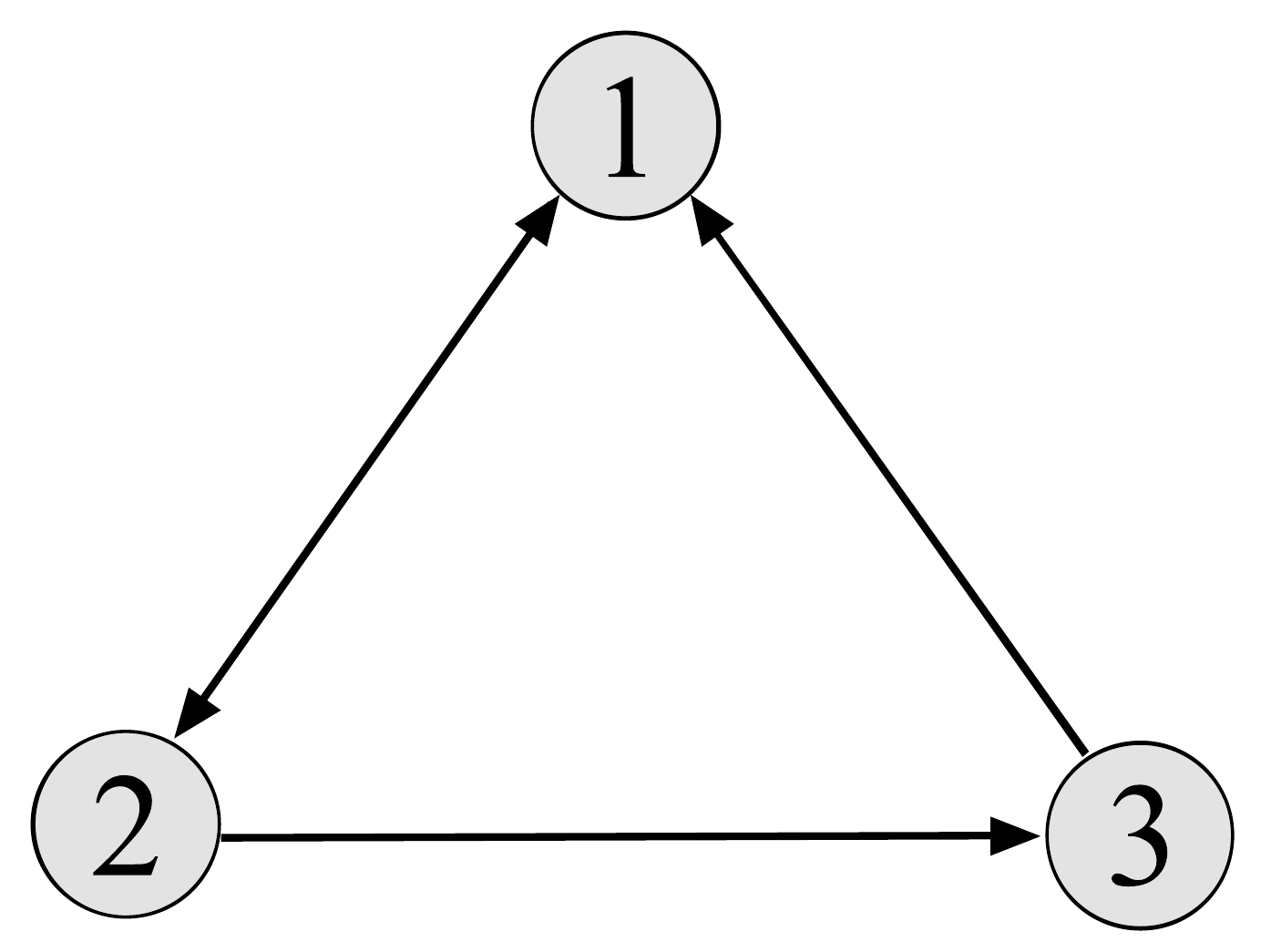}\\
  \end{minipage}
}
\hspace{1ex}
\subfloat[t][Autocratic network $\mathcal{G}(C^{3})$]{
\begin{minipage}[t]{0.29\textwidth}\label{f3-c}
 \centering
  \includegraphics[width=\hsize]{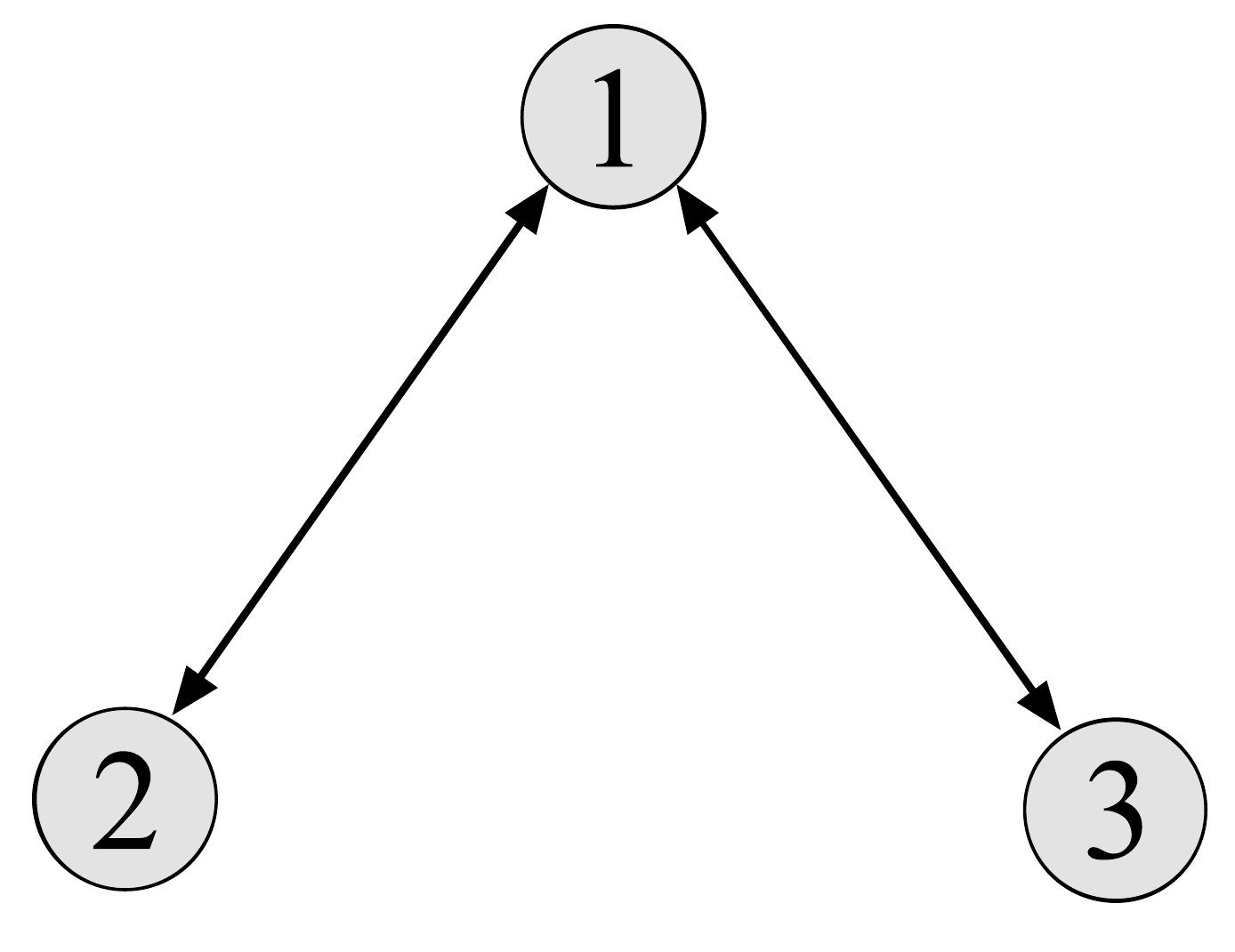}\\
  \end{minipage}
}
\vspace{-1ex}

\subfloat[t][$\sigma^{2}=(1, 2, 3)^{\top}$]{
\begin{minipage}[t]{0.4\textwidth}\label{f3-d}
 \centering
  \includegraphics[width=\hsize]{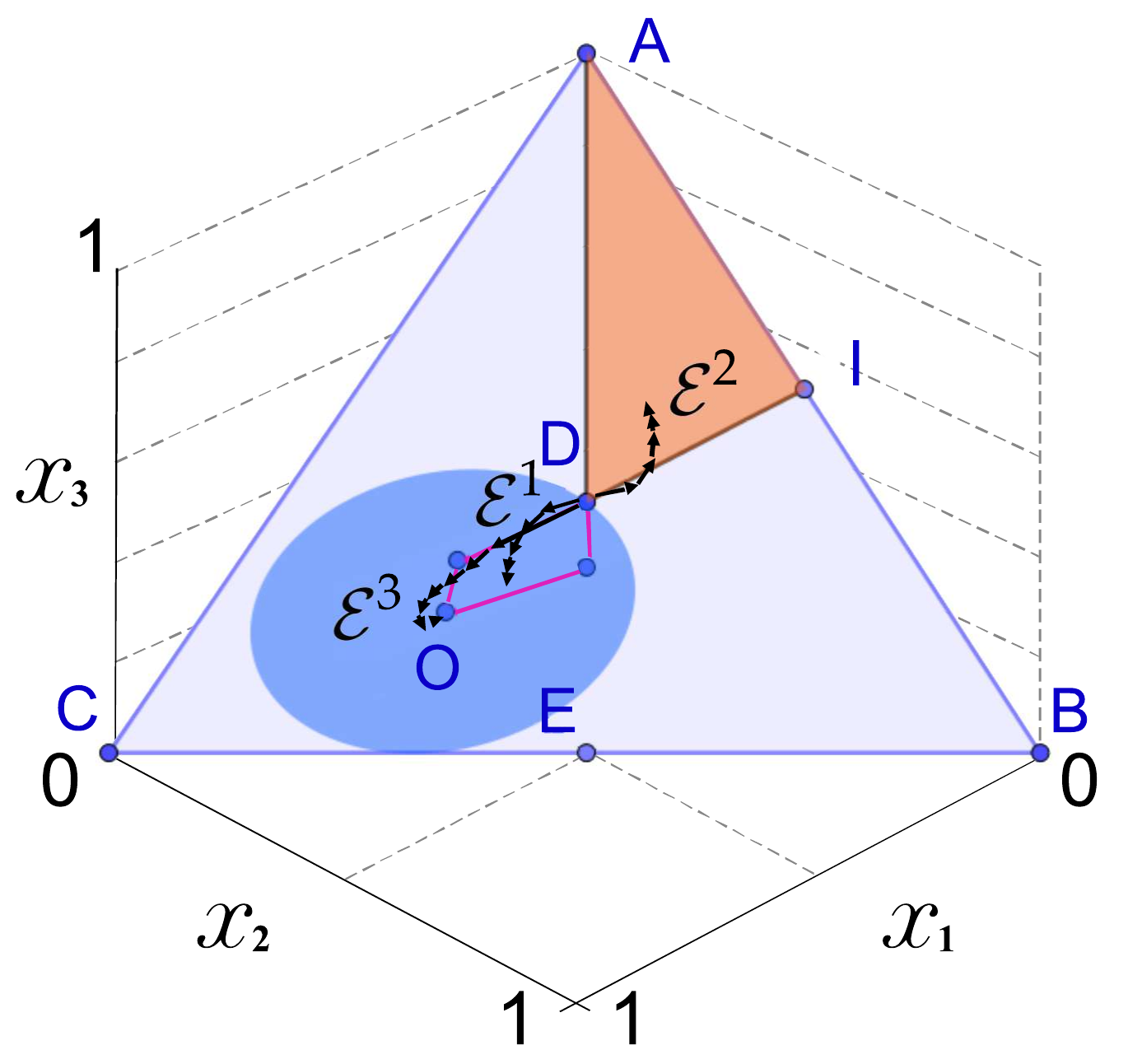}\\
 \end{minipage}
}
\hspace{3ex}
\subfloat[t][$\sigma^{2}=(1, 4, 9)^{\top}$]{
\begin{minipage}[t]{0.4\textwidth}\label{f3-e}
 \centering
  \includegraphics[width=\hsize]{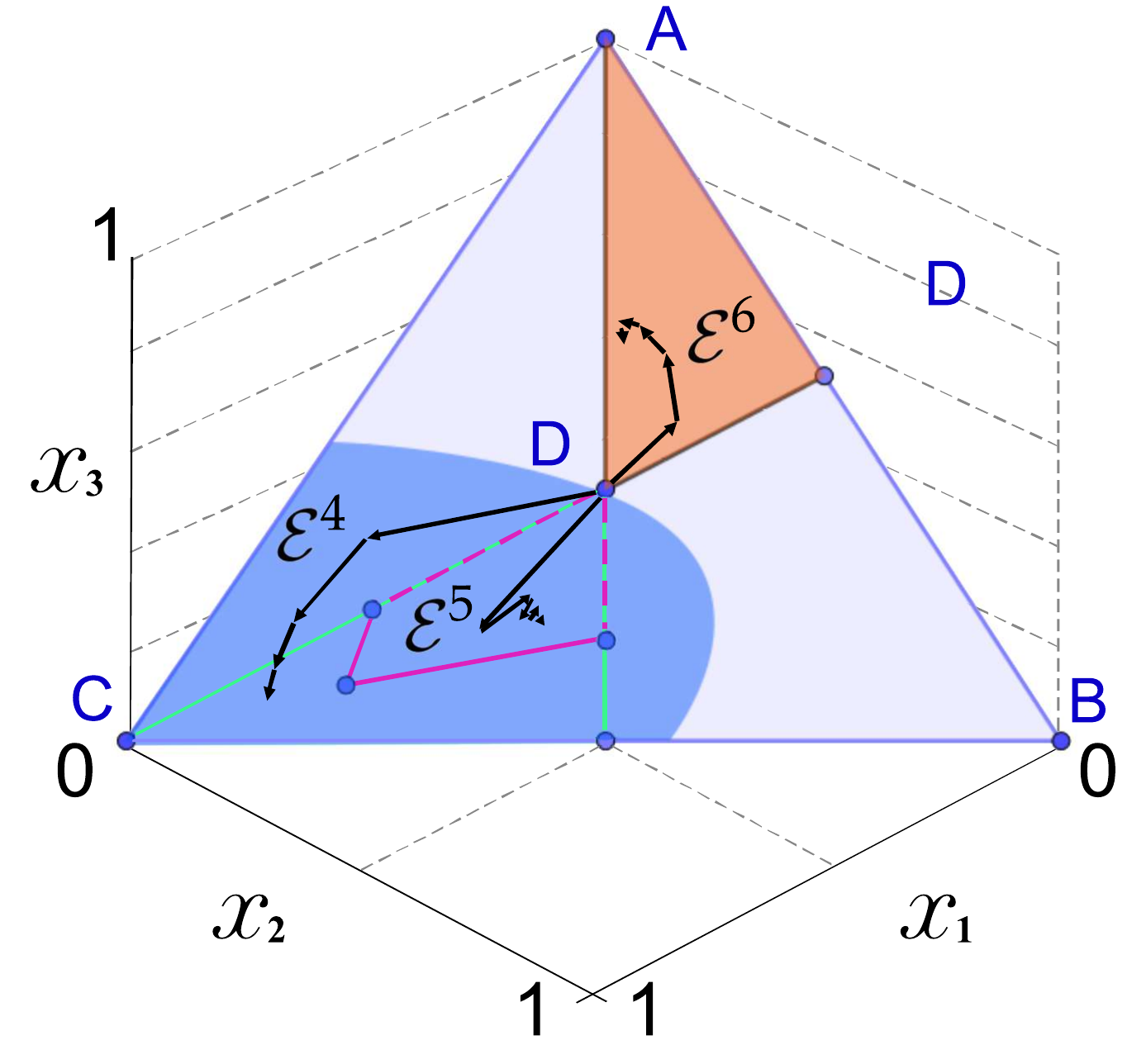}\\
 \end{minipage}
}
\caption{\Cref{f3-a,f3-b,f3-c} depict a democratic network $\mathcal{G}(C^{1})$, an irreducible network $\mathcal{G}(C^{2})$ and an autocratic network $\mathcal{G}(C^{3})$ with center node $1$, respectively, where $C^{1}=[0 \ 1 \ 0;0 \ 0 \ 1;1 \ 0 \ 0]$, $C^{2}=[0 \ 1 \ 0;0.5 \ 0 \ 0.5;1 \ 0 \ 0]$ and $C^{3}=[0 \ 0.3 \ 0.7;1 \ 0 \ 0;1 \ 0 \ 0]$. \Cref{f3-d,f3-e} depict the trajectories of the variances of the collective estimates $\Var[y_{
col}(k)]$ of the FD model~\eqref{e2} with different settings of individuals' variances, susceptibilities and relative interaction networks.}\label{fig3}
\end{figure}

\begin{remark}
In Delphi experiments, the relative interaction matrix $C$ is unknown to individuals and is designed by the inquiry moderator. A widely-adopted algorithm for designing relative influence network $\mathcal{G}(C)$ is the equal-neighbor model~\cite[Chapter~5]{FB:22}. Let $M$ be a binary adjacent matrix associated with connected undirected graph $\mathcal{G}$, i.e., $M_{ij}=M_{ji}=1$ if $(i,j)$ is an edge of $\mathcal{G}$, and $M_{ij}=M_{ji}=0$ otherwise. The corresponding equal-neighbor relative interaction matrix is $C=[M\mathbf{1}_{n}]^{-1}M$. Let $d=M\mathbf{1}_{n}$ be the degree vector of $\mathcal{G}$. It is well known that the left dominant eigenvector of $C$ is $c=\frac{1}{d^{\top}\mathbf{1}_{n}}(d_{1},\dots,d_{n})^{\top}$, where $M$ and $d$ is designed by the moderator. By~\eqref{e2}, 
\begin{align*}\label{e2.6}
\gamma_{i}=\frac{y_{i}(k+1)-y_{i}(k)}{\sum_{j=1}^{n}C_{ij}y_{j}(k)-y_{i}(k)},
\end{align*}
which means that the moderator can approximate an individual's susceptibility using its estimates at successive time step. Since $\frac{c_{i}}{c_{j}}=\frac{d_{i}}{d_{j}}$, \cref{C1} suggests that the moderator may dynamically assign node degrees to individuals according to their observed susceptibilities to improve the collective estimates.
\end{remark}

In the investigation of influence networks, democratic and autocratic networks are studied frequently. It is of interest that if $\mathcal{G}(C)$ is democratic, the left dominant eigenvector of $C$ (also known as eigenvector centrality scores) is $\mathbf{1}_{n}/n$; if $\mathcal{G}(C)$ is autocratic with center node $i$, the center node has the largest centrality score $1/2$, which is strictly larger than that of every non-center node~\cite{PJ-AM-NEF-FB:13d}. \Cref{f3-a,f3-c} depict a democratic network and an autocratic network with $3$ nodes, respectively.
\begin{corollary}[Improving/optimizing the wisdom in democratic/autocratic influence networks]\label{C2}
For system~\eqref{e2} with population $\mathcal{P}_{\sigma^{2},\gamma}$ and relative interaction network $\mathcal{G}(C)$, where $\gamma\in\Gamma=\setdef{z\in\mathbb{R}^{n}\setminus\{\mathbf{1}_{n}\}}{1\geq z_{i}>0, i\in\until{n}}$, 
\begin{enumerate}
\item suppose that $\mathcal{G}(C)$ is democratic, then\label{C2-1}
\begin{enumerate}
\item system~\eqref{e2} asymptotically optimizes the wisdom if and only if $\mathcal{P}_{\sigma^{2},\gamma}$ is a maximal PAP population; \label{C2-1-a}

\item if $\sigma^{2}\notin\Span{\mathbf{1}_{n}}$ and $\mathcal{P}_{\sigma^{2},\gamma}$ is a strong PAP population, system~\eqref{e2} asymptotically improves the wisdom; 

\item assume $\sigma_{i}^{2}\leq\sigma_{i+1}^{2}$ for all $i\in\until{n-1}$. System~\eqref{e2} with any PAP population asymptotically improves the wisdom if and only if individuals admit the $\tau^{0}$-hierarchy.

\end{enumerate}
\item Suppose that $\mathcal{G}(C)$ is autocratic with center node $l$, $C_{li}$ is the influence weight the center node accords to $i$, then \label{C2-2}
\begin{enumerate}
\item if $\sigma^{2}\notin\Span{\mathbf{1}_{n}}$, $\frac{\gamma_{i}}{\gamma_{l}}\frac{\sigma_{l}^{2}}{\sigma^{2}_{i}}\leq C_{li}<\frac{\gamma_{i}}{\gamma_{l}}$ for all $\sigma_{i}^{2}\geq\sigma_{l}^{2}$, $\frac{\gamma_{i}}{\gamma_{l}}\frac{\sigma_{l}^{2}}{\sigma^{2}_{i}}\geq C_{li}>\frac{\gamma_{i}}{\gamma_{l}}$ for all $\sigma_{i}^{2}\leq\sigma_{l}^{2}$ and $\frac{\gamma_{i}}{\gamma_{j}}\frac{\sigma_{j}^{2}}{\sigma^{2}_{i}}\leq\frac{C_{li}}{C_{lj}}\leq\frac{\gamma_{i}}{\gamma_{j}}$ for all $\sigma_{i}^{2}\geq\sigma_{j}^{2}$, $i,j\neq l$, system~\eqref{e2} asymptotically improves the wisdom; \label{C2-2-a}

\item system~\eqref{e2} asymptotically optimizes the wisdom asymptotically if and only if $\frac{\gamma_{i}}{\gamma_{l}}=C_{li}\frac{\sigma^{2}_{i}}{\sigma^{2}_{l}}$ for all $i\neq l$. 
\end{enumerate}
\end{enumerate}
\end{corollary}

\Cref{C2}~\ref{C2-1} shows that the \emph{PAP hypothesis} is significant in democratic influence networks. If more accurate individuals are more resistant to social influence in some extent, the wisdom is improved or optimized in democratic influence networks. \Cref{C2}~\ref{C2-2} shows that the FD opinion dynamics can improve or optimize wisdom even in autocratic influence networks. $\mathcal{E}^3$ and $\mathcal{E}^5$ in \cref{f3-d,f3-e} are the trajectories of the variances of collective estimates of system~\eqref{e2} under democratic network $\mathcal{G}(C^1)$ and autocratic network  $\mathcal{G}(C^3)$ with centralities scores $c=(1/3, 1/3, 1/3)^{\top}$ and $c=(0.5, 0.15, 0.35)^{\top}$, respectively. Correspondingly, individuals' susceptibilities are $\gamma=(0.2, 0.4, 0.6)^{\top}$ and $\gamma=(0.4, 0.2, 0.9)^{\top}$, which satisfy the conditions in \cref{C2}~\ref{C2-1-a} and \cref{C2}~\ref{C2-2-a}, respectively. Thus, the wisdom is optimized and improved respectively.

\subsection{Undermining the wisdom in irreducible influence networks}
We end this section by summarizing results on undermining the wisdom with the FD opinion dynamics in irreducible influence networks. 
\begin{corollary}[Undermining the wisdom in irreducible influence networks]\label{C3}
For system~\eqref{e2} with population $\mathcal{P}_{\sigma^{2},\gamma}$ and relative interaction network $\mathcal{G}(C)$, where $\gamma\in\Gamma=\setdef{z\in\mathbb{R}^{n}\setminus\{\mathbf{1}_{n}\}}{1\geq z_{i}>0, i\in\until{n}}$, $C$ is irreducible with dominant left eigenvector $c\in\interior\Delta_{n}$, 
\begin{enumerate}
\item if $\frac{c_{i}}{c_{j}}\geq\frac{\gamma_{i}}{\gamma_{j}}$ for all $\sigma_{i}^{2}\geq\sigma_{j}^{2}$ with at least one inequality holding strictly, system~\eqref{e2} asymptotically undermines the wisdom;\label{C3-1}

\item assume that individuals admit a $\tau$-hierarchy, then system~\eqref{e2} asymptotically undermines the wisdom if $\frac{c_{\tau_{i}}}{c_{\tau_{i+1}}}\leq\frac{\gamma_{\tau_{i}}}{\gamma_{\tau_{i+1}}}$ for all $i\in\until{n-1}$ with at least one strict inequality holding. \label{C3-2}

\end{enumerate}

\end{corollary}

Comparing \cref{C3} with \cref{C1}, if more accurate individuals have less network centralities and are more susceptible to social influence, the wisdom is undermined; on the contrary, if more accurate individuals have larger network centralities and are more resistant to social influence, the wisdom is improved. $\mathcal{E}^2$ and $\mathcal{E}^6$ in \cref{f3-d,f3-e} are the trajectories of the variances of collective estimates of system~\eqref{e2} under irreducible network $\mathcal{G}(C^2)$ and democratic network $\mathcal{G}(C^1)$, respectively. Individuals' susceptibilities are $\gamma=(0.4, 0.3, 0.1)^{\top}$ and $\gamma=(0.6, 0.5, 0.2)^{\top}$, which satisfy the conditions in \cref{C3}~\ref{C3-1} and \cref{C3}~\ref{C3-2}, respectively. Thus, system~\eqref{e2} undermines the wisdom under both settings.

\section{Conclusions}\label{S6}
This paper has investigated the problem of improving, optimizing and
undermining the collective wisdom in influence networks. Mathematical
formulation and rigorous analysis are provided. Various results for
improving, optimizing and undermining of the wisdom are derived from the
perspective of influence systems theory. Our theoretical results contribute
to the study of influence networks and the wisdom of crowds effect in the
following aspects: first, we show that wisdom in influence networks can be
both improved and undermined by social influence. Second, we provide
necessary and/or sufficient conditions for improving, optimizing and
undermining the wisdom. Finally, we provide theoretical explanations for
empirical evidence reported in the literature, such as the \emph{PAP
hypothesis}. 

Our investigation provides a mathematical perspective on the debate about how social influence affects the wisdom of crowds. Our approach is based upon influence system theory. As the first mathematical modelling step, this paper does not rely upon a particular dynamical model for individuals to interact and evolve their estimates. Instead, our analysis is based directly upon social power allocations in the influence system. On one hand, this modelling assumption ensures the universality of our results. On the other hand, this modelling assumption does not take into account the dynamic nature of social networks and any resulting dynamic adjustments of social power allocations. For example, individuals may dynamically modify the
influence weights they accord to others or adaptively adjust their
susceptibilities to social influence. At this time it is unclear what
information and what mechanisms are dominant in social power evolution;
this is a modelling problem with a corresponding analysis problem. We leave
these topics to further investigations.

\appendix
\section{Proof of \cref{L1}}\label{AP1}
Since $\sigma^{2}_{i}\leq\beta<\infty$ for all $i\in\until n$, we have
\begin{align*}
\lim_{n\to \infty}\Var[y_{\col}^{(n)}(0)]=\lim_{n\to \infty}\frac{1}{n^2}\sum_{i=1}^{n}\sigma_{i}^{2}
\leq\lim_{n\to \infty}\frac{\beta}{n}=0.
\end{align*}
Note that social influence asymptotically improves the wisdom, i.e., for all $n\geq 2$ we have $\Var[\lim_{k\to\infty}y_{\col}^{(n)}(k)]<\Var[y_{\col}^{(n)}(0)]$ for all $n\geq 2$. By Chebyshev's inequality \cite[Theorem 1.6.4]{RD:10}, we obtain
\begin{align*}
\Prob[\mid \lim_{k\to\infty}y_{\col}^{(n)}(k)-\mu\mid\geq\epsilon]\leq\frac{\Var[\lim_{k\to\infty}y_{\col}^{(n)}(k)]}{\epsilon^{2}}
<\frac{\Var[y_{\col}^{(n)}(0)]}{\epsilon^{2}}
\end{align*}
for any $\epsilon>0$. Since $\lim_{n\to \infty}\Var[y_{\col}^{(n)}(0)]=0$, we obtain $\lim_{n\to\infty}\Prob[\lim_{k\to\infty}y_{\col}^{(n)}(k)=\mu]=1$, which means the final collective estimate asymptotically converges to the truth {\it i.p.} as $n\to\infty$.
\section{Proof of \cref{L3}}\label{AP2}
Regarding~\ref{L3-0}, let $0<\lambda<1$. For any $z, z^{\prime}\in\Delta_{n}$ and $z\neq z^{\prime}$, we have
\begin{align*}
&\lambda \mathcal{E}_{\sigma^{2}}(z)+(1-\lambda)\mathcal{E}_{\sigma^{2}}(z^{\prime})-\mathcal{E}_{\sigma^{2}}(\lambda z+(1-\lambda)z^{\prime})\\
=&\lambda\!\sum_{i=1}^{n}z^{2}_{i}\sigma_{i}^{2}\!+\!(1\!-\!\lambda)\!\sum_{i=1}^{n}(z^{\prime}_{i})^{2}\sigma_{i}^{2}\!-\!\sum_{i=1}^{n}(\lambda z_{i}\!+\!(1\!-\!\lambda)z^{\prime}_{i})^{2}\sigma_{i}^{2}\\
=&\lambda(1-\lambda)(\sum_{i=1}^{n}z^{2}_{i}\sigma_{i}^{2}+\sum_{i=1}^{n}(z^{\prime}_{i})^{2}\sigma_{i}^{2}-2\sum_{i=1}^{n}z_{i}z^{\prime}_{i}\sigma_{i}^{2})\\
=&\lambda(1-\lambda)\sum_{i=1}^{n}(z_{i}-z^{\prime}_{i})^{2}\sigma_{i}^{2}>0.
\end{align*} 
Thus, $\mathcal{E}_{\sigma^{2}}(z)$ is strictly convex on $\Delta_{n}$. Moreover, for any $z,z^{\prime}\in\mathcal{A}_{\sigma^{2}}$, there holds
\begin{align*}
\mathcal{E}_{\sigma^{2}}(\lambda z+(1-\lambda)z^{\prime})<\lambda \mathcal{E}_{\sigma^{2}}(z)+(1-\lambda)\mathcal{E}_{\sigma^{2}}(z^{\prime})
<\mathcal{E}_{\sigma^{2}}(\frac{\mathbf{1}_{n}}{n}),
\end{align*} 
where the first inequality is implied by strict convexity of $\mathcal{E}_{\sigma^{2}}(z)$ on $\Delta_{n}$ and the second inequality follows from $z,z^{\prime}\in\mathcal{A}_{\sigma^{2}}$. Therefore, $\mathcal{A}_{\sigma^{2}}$ is convex.

Regarding~\ref{L3-1}, equivalently, we prove $\mathcal{A}_{\sigma^{2}}=\emptyset$ if and only if $\sigma^{2}\in\Span{\mathbf{1}_{n}}$. {\it Necessity}. If $\mathcal{A}_{\sigma^{2}}=\emptyset$, then $\mathcal{E}_{\sigma^{2}}(z)\geq\mathcal{E}_{\sigma^{2}}(\frac{\mathbf{1}_{n}}{n})$ for all $z\in\Delta_{n}$. Suppose that there exist $i,j\in\until{n}$ such that $\sigma_{i}^{2}>\sigma_{j}^{2}$. For $\alpha\in(0,\frac{1}{n})$, note that $\frac{\mathbf{1}_{n}}{n}-\alpha\mathbf{e}_{i}+\alpha\mathbf{e}_{j}\in\Delta_{n}$, where $\mathbf{e}_{i}\in\mathbb{R}^{n}$ is the $i$-th canonical basis. Moreover,  
\begin{align*}
\mathcal{E}_{\sigma^{2}}(\frac{\mathbf{1}_{n}}{n}-\alpha\mathbf{e}_{i}+\alpha\mathbf{e}_{j})-\mathcal{E}_{\sigma^{2}}(\frac{\mathbf{1}_{n}}{n})
=&(\frac{1}{n}-\alpha)^{2}\sigma_{i}^{2}+(\frac{1}{n}+\alpha)^{2}\sigma_{j}^{2}-\frac{1}{n^{2}}\sigma_{i}^{2}-\frac{1}{n^{2}}\sigma_{j}^{2}\\
=&\alpha(\alpha(\sigma_{i}^{2}+\sigma_{j}^{2})-\frac{2}{n}(\sigma_{i}^{2}-\sigma_{j}^{2})),
\end{align*}
in which $\alpha(\sigma_{i}^{2}+\sigma_{j}^{2})-\frac{2}{n}(\sigma_{i}^{2}-\sigma_{j}^{2})<0$ for $\alpha<\min\{\frac{1}{n}, \frac{2(\sigma_{i}^{2}-\sigma_{j}^{2})}{n(\sigma_{i}^{2}+\sigma_{j}^{2})}\}$, and the existence of such $\alpha$ is guaranteed by $\sigma_{i}^{2}>\sigma_{j}^{2}$. That is to say, there exists $\frac{\mathbf{1}_{n}}{n}-\alpha\mathbf{e}_{i}+\alpha\mathbf{e}_{j}\in\Delta_{n}$ such that $\mathcal{E}_{\sigma^{2}}(\frac{\mathbf{1}_{n}}{n}-\alpha\mathbf{e}_{i}+\alpha\mathbf{e}_{j})<\mathcal{E}_{\sigma^{2}}(\frac{\mathbf{1}_{n}}{n})$, which is contradicted with $\mathcal{A}_{\sigma^{2}}=\emptyset$. Therefore, $\sigma^{2}\in\Span{\mathbf{1}_{n}}$ if $\mathcal{A}_{\sigma^{2}}=\emptyset$.

{\it Sufficiency}. Let $\sigma^{2}=\delta\mathbf{1}_{n}$ with $\delta\in (0,\infty)$. Then, for any $z\in\Delta_{n}$, there holds $\mathcal{E}_{\sigma^{2}}(z)=\delta\sum_{i=1}^{n}z_{i}^{2}$. Since $\argmin_{z\in\Delta_{n}}\sum_{i=1}^{n}z_{i}^{2}=\frac{\mathbf{1}_{n}}{n}$, we have $\min_{z\in\Delta_{n}}\mathcal{E}_{\sigma^{2}}(z)=\mathcal{E}_{\sigma^{2}}(\frac{\mathbf{1}_{n}}{n})$, which means $\mathcal{A}_{\sigma^{2}}=\emptyset$.

Regarding~\ref{L3-2}, since $\mathcal{\tilde{A}}_{\sigma^{2}}\subset\mathcal{A}_{\sigma^{2}}$ if $\mathcal{\tilde{A}}_{\sigma^{2}}=\emptyset$, we just need to consider the case that $\mathcal{\tilde{A}}_{\sigma^{2}}\neq\emptyset$. First, if $\sigma^{2}\in\Span{\mathbf{1}_{n}}$, then $\mathcal{\tilde{A}}_{\sigma^{2}}=\emptyset$ since $\mathcal{\tilde{A}}_{\sigma^{2}}\subset\interior{ \Delta_{n}}\setminus \{\frac{\mathbf{1}_{n}}{n}\}$; if $\sigma^{2}\notin\Span{\mathbf{1}_{n}}$ and there exists $i\in\until{n-1}$ such that $\sigma_{i}^{2}>\sigma_{i+1}^{2}$, then $\mathcal{\tilde{A}}_{\sigma^{2}}=\emptyset$. Thus, $\mathcal{\tilde{A}}_{\sigma^{2}}\neq\emptyset$ only if $\sigma^{2}\notin\Span{\mathbf{1}_{n}}$ and $\sigma_{i}^{2}\leq\sigma_{i+1}^{2}$ for all $i\in\until{n-1}$. Note that $\mathcal{E}_{\sigma^{2}}(z)$ is differentiable on $\Delta_{n}$, let $\nabla \mathcal{E}_{\sigma^{2}}(z)\in\mathbb{R}^{n}$ be its gradient. For all $z\in\mathcal{\tilde{A}}_{\sigma^{2}}$, the mean value theorem~\cite[Theorem 3.2.2]{JMO-WCR:70} ensures the existence of $z^{\prime}=\lambda z+(1-\lambda)\frac{\mathbf{1}_{n}}{n}$ with $\lambda\in(0,1)$ such that 
\begin{align*}
\mathcal{E}_{\sigma^{2}}(z)-\mathcal{E}_{\sigma^{2}}(\frac{\mathbf{1}_{n}}{n})=\nabla^{\top}\mathcal{E}_{\sigma^{2}}(z^{\prime})(z-\frac{\mathbf{1}_{n}}{n}).
\end{align*}
Because $z\neq \frac{\mathbf{1}_{n}}{n}$ and $z_{i}\geq z_{i+1}$ for all $i\in\until{n-1}$, there exists $l\in\until{n-1}$ such that $z_{i}>\frac{1}{n}$ for $i\leq l$ and $z_{i}\leq\frac{1}{n}$ for $i>l$. Note that $\frac{\partial\mathcal{E}_{\sigma^{2}}(z)}{\partial z_{i}}=2z_{i}\sigma_{i}^{2} $, thus 
\begin{align*}
\mathcal{E}_{\sigma^{2}}(z)-\mathcal{E}_{\sigma^{2}}(\frac{\mathbf{1}_{n}}{n})
=2\sum_{i=1}^{l}z^{\prime}_{i}\sigma_{i}^{2}(z_{i}-\frac{1}{n})-2\sum_{i=l+1}^{n}z^{\prime}_{i}\sigma_{i}^{2}(\frac{1}{n}-z_{i}),
\end{align*}
where $\sum_{i=1}^{l}(z_{i}-\frac{1}{n})=\sum_{i=l+1}^{n}(\frac{1}{n}-z_{i})>0$ because $z\in\Delta_{n}$. Since $\sigma_{i}^{2}\leq\sigma_{i+1}^{2}$ and $z_{i}\sigma_{i}^{2}\leq z_{i+1}\sigma_{i+1}^{2}$ for all $i\in\until{n-1}$, we obtain 
\begin{align*}
z^{\prime}_{i}\sigma_{i}^{2}=\lambda z_{i}\sigma_{i}^{2}+(1-\lambda)\frac{1}{n}\sigma_{i}^{2}
\leq \lambda z_{i+1}\sigma_{i+1}^{2}+(1-\lambda)\frac{1}{n}\sigma_{i+1}^{2}=z^{\prime}_{i+1}\sigma_{i+1}^{2}
\end{align*}
for all $i\in\until{n-1}$. Particularly, $\sigma_{l}^{2}<\sigma_{l+1}^{2}$ is implied by $z_{l}>z_{l+1}$ and $z_{l}\sigma_{l}^{2}\leq z_{l+1}\sigma_{l+1}^{2}$, and implies $z^{\prime}_{l}\sigma_{l}^{2}<z^{\prime}_{l+1}\sigma_{l+1}^{2}$. Therefore, we obtain 
\begin{align*}
\mathcal{E}_{\sigma^{2}}(z)-\mathcal{E}_{\sigma^{2}}(\frac{\mathbf{1}_{n}}{n})
\leq 2(z^{\prime}_{l}\sigma_{l}^{2}-z^{\prime}_{l+1}\sigma_{l+1}^{2})\sum_{i=1}^{t}(z_{i}-\frac{1}{n})<0,
\end{align*}
which means $z\in\mathcal{A}_{\sigma^{2}}$. Thus, $\mathcal{\tilde{A}}_{\sigma^{2}}\subset\mathcal{A}_{\sigma^{2}}$.

\section{Proof of \cref{L4}}\label{AP3}
Regarding~\ref{L4-1}. {\it Necessity}. Suppose there exists $j\in\until{n-1}$ such that $\frac{1}{j^2}\sum_{r=1}^{j}\sigma_{r}^{2}\geq\frac{1}{n^2}\sum_{i=1}^{n}\sigma_{i}^{2}$. Let $\hat{z}=\frac{1}{j}\sum_{i=1}^{j}\mathbf{e}_{i}$, where $\mathbf{e}_{i}\in\mathbb{R}^{n}$ is the $i$-th canonical basis. Then, $\hat{z}_{i}=\frac{1}{j}$ for all $i\leq j$ and $\hat{z}_{i}=0$ for all $i>j$, that is to say, $\hat{z}\in\Delta^{\tau^{0}}_{n}$. On the other hand,
\begin{align*}
\mathcal{E}_{\sigma^{2}}(\hat{z})=\frac{1}{j^2}\sum_{r=1}^{j}\sigma_{r}^{2}\geq\frac{1}{n^2}\sum_{i=1}^{n}\sigma_{i}^{2}.
\end{align*}
Hence, $\hat{z}\not\in \mathcal{A}_{\sigma^{2}}$, thus $\Delta^{\tau^{0}}_{n}\not\subset \mathcal{A}_{\sigma^{2}}$, which is a contradiction. Therefore, $\Delta^{\tau^{0}}_{n}\subset \mathcal{A}_{\sigma^{2}}$ only if~\eqref{eq2} holds for all $j\in\until{n-1}$.  

{\it Sufficiency}. Denote $\tilde{\mathbf{e}}_{i}=\frac{1}{i}\sum_{r=1}^{i}\mathbf{e}_{r}\in\Delta_{n}$. First, we show that for all $z\in\Delta^{\tau^{0}}_{n}$, there exists $q\in\Delta_{n}$ such that $z=\sum_{i=1}^{n}q_{i}\tilde{\mathbf{e}}_{i}$. Let $\delta_{i}=z_{i}-z_{i+1}$, $i\in\until{n-1}$ and $\delta_{n}=z_{n}$. Note that $\delta_{i}\geq 0$ and there exists at least one $\delta_{i}>0$ since $z\neq\frac{\mathbf{1}_{n}}{n}$. Let $q_{i}=i\delta_{i}$ for all $i\in\until{n}$. Then, we have $q\geq 0$, 
\begin{align*}
\sum_{i=1}^{n}q_{i}=\sum_{i=1}^{n-1}i(z_{i}-z_{i+1})+nz_{n}=1,
\end{align*}
and
\begin{align*}
\sum_{i=1}^{n}q_{i}\tilde{\mathbf{e}}_{i}=\sum_{i=1}^{n}\delta_{i}\sum_{r=1}^{i}\mathbf{e}_{r}=\sum_{i=1}^{n}\mathbf{e}_{i}\sum_{r=i}^{n}\delta_{r}=z,
\end{align*}
where the last equality follows from $\sum_{r=i}^{n}\delta_{r}=z_{i}$. Recall that $\mathcal{E}_{\sigma^{2}}(z)$ is strictly convex on $\Delta_{n}$, and $\tilde{\mathbf{e}}_{i}\in\Delta_{n}$ for all $i\in\until{n}$, we obtain 
\begin{align*}
\mathcal{E}_{\sigma^{2}}(z)=\mathcal{E}_{\sigma^{2}}(\sum_{i=1}^{n}q_{i}\tilde{\mathbf{e}}_{i})<\sum_{i=1}^{n}q_{i}\mathcal{E}_{\sigma^{2}}(\tilde{\mathbf{e}}_{i})<\frac{1}{n^2}\sum_{i=1}^{n}\sigma_{i}^{2},
\end{align*}
where the last inequality follows from $\sum_{i=1}^{n}q_{i}=1$ and $\mathcal{E}_{\sigma^{2}}(\tilde{\mathbf{e}}_{j})=\frac{1}{j^2}\sum_{r=1}^{j}\sigma_{r}^{2}<\frac{1}{n^2}\sum_{i=1}^{n}\sigma_{i}^{2}$ for all $j\in\until{n-1}$. 

The proof of statement~\ref{L4-2} is similar to the proof of statement~\ref{L4-1}. We omit the proof by noticing the facts that $\mathcal{\hat{A}}_{m}\subset \Delta^{\tau^{0}}_{n}$ and $z=\sum_{i=1}^{m}q_{i}\tilde{\mathbf{e}}_{i}$ for all $z\in\mathcal{\hat{A}}_{m}$ with $q_{i}$ and $\tilde{\mathbf{e}}_{i}$ defined above.

\bibliographystyle{siamplain}
\bibliography{alias,Main,FB}

\end{document}